\documentclass{ws-ijcga}

\usepackage{graphicx,tikz}
\usetikzlibrary{shapes.geometric, arrows}
\usepackage[linesnumbered,boxed]{algorithm2e}
\usepackage{caption,amssymb,amsmath}

\begin{document}

\markboth{Kai Jin}
{Finding all Maximal Area Parallelograms in a Convex Polygon}

\catchline

\title{Finding all Maximal Area Parallelograms in a Convex Polygon\footnote{The conference version of this paper was published in CCCG 2011: The 23rd Canadian Conference on Computational Geometry. Moreover, at the date of submitting this manuscript, there is already a followup research which improves this result.
The followup research combines this part of research with another part which introduces and studies a novel geometric structure.
Since it is way too lengthy to put both parts inside one single paper, we decide to publish two parts sequentially.}}

\author{Kai Jin\footnote{Supported by National Basic Research Program of China Grant 2007CB807900, 2007CB807901, and National Natural Science Foundation of China Grant 61033001, 61061130540, 61073174.}}

\address{Department of Computer Science, The University of Hong Kong,\\
Pokfulam Road, Hong Kong SAR, China. \\ \texttt{cscjjk@gmail.com}
}

\maketitle

\pub{Received (received date)}{Revised (revised date)}
{Communicated by (Name)}

\begin{abstract}
Polygon inclusion problems have been studied extensively in geometric optimization.
In this paper, we consider the variant of computing the maximum area parallelograms (MAPs) and
  all the locally maximal area parallelograms (LMAPs) in a given convex polygon.
By proving and utilizing several structural properties of the LMAPs, we compute all of them (including all the MAPs)
  in $O(n^2)$ time, where $n$ denotes the number of edges of the given polygon.
In addition, we prove that the LMAPs interleave each other and thus the number of LMAPs is $O(n)$.
We discuss applications of our result to, among others, the problem of computing
  the maximum area centrally-symmetric convex body inside a convex polygon,
  and the simplest case of the Heilbronn triangle problem.

\end{abstract}

\section{Introduction}\label{sect:introduction}

Polygon inclusion problems consider searching for extremal shapes with special properties inside a polygon and have been studied for four decades.
Famous pioneer work includes the studies of the diameter problem \cite{Classic-shamosCG-dissertation},
  potato-peeling problem \cite{Potato-polynomial-DCG86} (which concerns finding the largest convex polygon in a simple polygon),
  maximum $k$-gon problem \cite{Kgon-extremal-STOC82},
  maximum ellipse problem \cite{ellipse-LPtype-SCG92},
  maximum equilateral triangle and squares problem \cite{Othershape-square-Allerton87},
  maximum homothetic / similar copy problem \cite{Placement-chazelle-CR83,Placement-ST-CGTA94,Placement-convex-DCG98}.
  See subsection~\ref{subsect:related} for an overview of these and other results.
  These problems usually arise from more practical problems of operations research and pattern recognition \cite{ShapeSurvey}. \smallskip

In this paper, we study the following polygon inclusion problem:
  \emph{Given a convex polygon $P$ bounded by $n$ halfplanes, compute the maximum area parallelograms (MAPs) in $P$.}
In addition to the fact that this problem is as clean and fundamental as the aforementioned related problems, it has several special motivations:

\smallskip \noindent\textbf{I.} \mbox{  } It is a natural extension of the diameter problem.
  To find the diameter, we determine a center $O$ and one vector $\beta_1$ so that $O\pm \beta_1$ are contained in $P$
    and that a measure (i.e.\ length) of the vector is maximized.
  To find the MAP, we determine a center $O$ and two vectors $\beta_1,\beta_2$ so that $O\pm \beta_1,O\pm \beta_2$ are contained in $P$
    and that a measure (i.e.\ cross product) of the vectors is maximized.
    (By this observation, our problem can be formulated as a quadratic programming; see \ref{sect:address-cp}.)\smallskip

\smallskip \noindent\textbf{II.} \mbox{  } By finding the MAP in $P$, we find an affine transformation $\sigma$ in the special linear group $SL(2)$,
    so that the area of the largest \emph{square} in $\sigma(P)$ is maximized.
    We can bring the polygon into a ``good position'' by transformation $\sigma$, to avoid almost degenerate, i.e., needle-like bodies.
    Alternatively, we may find the maximum area ellipse (MAE) and thus find $\sigma'$ so that the area of the largest \emph{circle} in $\sigma'(P)$ is maximized,
    yet the algorithms for finding the MAE are not as practical as ours.\smallskip

\smallskip \noindent\textbf{III.} \mbox{  } Because parallelograms are the simplest polygons that are centrally symmetric, the MAPs are good approximations of the maximum area centrally symmetric body (MAC) in $P$.
We may give the simplest credit to squares, but obeying only one constraint of being centrally-symmetric,
  parallelograms are simpler in the less-constraints sense and hence are more suitable for approximating the MAC.
\begin{theorem}\cite{Math-Sas-CM39,Math-Dow-AMS44}\label{theorem:approx}
For any centrally symmetric convex region $K$, the area of the MAP in $K$ is at least $\frac{2}{\pi}$ times the area of $K$.
  This bound is tight; $\frac{2}{\pi}$ is realized when $K$ is surrounded by an ellipse.
  As a corollary, computing the MAP serves as a $\frac{2}{\pi}$-approximation for computing the MAC in a convex region.
\end{theorem}

See the literature of the above theorem in \ref{sect:ratio}.\smallskip

\smallskip \noindent\textbf{IV.} \mbox{  }  By finding the MAPs, we solve the easiest case of the Heilbronn triangle problem, a minimax problem in discrete geometry. This is discussed in subsection~\ref{subsect:heil}.

\medskip \noindent\textbf{Our results.}
   A locally maximal area parallelograms (LMAP) is a parallelogram in $P$ that is no smaller than all its sufficiently close parallelograms that lie in $P$ (see Definition~\ref{def:LMAP}). Our results are summarized in the following.

   \smallskip 1.  We prove a group of $\Theta(n^2)$ nontrivial constraints on the corners of the LMAPs.
      These constraints are of independent interest in convex geometry --
       as shown in our subsequent manuscript \cite{followup-MAP-arxiv15}, an interesting geometric structure associated with the given convex polygon $P$ can be defined from these constraints.

   \smallskip 2. By utilizing these constraints, we design an $O(n^2)$ time algorithm for computing all the LMAPs (and thus computing all the MAPs).

   \smallskip 3. As an auxiliary result, we prove that the LMAPs are pairwise interleaving (see Definition~\ref{def:interleaving}),
    which implies that the number of LMAPs is only $O(n)$.

\subsection{Related work and some comparisons}\label{subsect:related}

As a typical example of the polygon inclusion problem, the maximum area / perimeter $k$-gon in a convex polygon has been studied intensively.
The first algorithm for computing it runs in $O(kn\log n+n \log^2n)$ time \cite{Kgon-extremal-STOC82}.
A factor of $\log n$ is saved later by the \emph{matrix search technique} \cite{Kgon-Matrixsearch-Algc87} (see also \cite{Kgon-Klink-DCG94} and \cite{Kgon-Klink-soda95}).
Chandran and Mount \cite{Triangle-correct-IJCGA02} presented a linear time algorithm for the special case of $k=3$, but only for the maximum area case.
An alternative algorithm is found later by Jin \cite{Triangle-ultimate-Arxiv}, which also runs in linear time but is more elegant and simpler.
An earlier linear time algorithm claimed in \cite{Triangle-wrong-FOCS79} is wrong due to a report of \cite{Triangle-reportwrong-Arxiv}.
For the maximum perimeter triangle, the best known algorithm is given in \cite{Kgon-extremal-STOC82}, which runs in $O(n\log n)$ time.
Note that the corners of the maximum $4$-gon (or any $k$-gon) can be restricted to the vertices of the given polygon $P$,
  whereas the corners of the MAP can only be restricted to the boundary of $P$ (proved in Lemma~\ref{lemma:local-maximal-non-slidable} below),
    the MAP problem might be inherently more difficult than the maximum $4$-gon problem.
    So, it could be acceptable that our algorithm is not as efficient as the one for computing the maximum $4$-gon.

In fact, our algorithm is as efficient as the best algorithms for many related problems, e.g.\
    the $O(n^3)$ time algorithm for computing the maximum rectangles \cite{Othershape-rect-EuroCG14},
      the $O(n^2\log n)$ time one for the maximum similar copy of a given triangle \cite{Placement-ST-CGTA94},
      and the $O(n^2)$ time one for the maximum inscribed squares / equilateral triangles \cite{Othershape-square-Allerton87}.

\smallskip Researchers also search for the extremal shapes enclosing a given convex polygon.
Interestingly, both the algorithm in \cite{Triangle-correct-IJCGA02} and the one in \cite{Triangle-ultimate-Arxiv} mentioned above can also find the minimum area enclosing triangle in linear time.
Yet the first linear time algorithm for computing such a triangle is given in \cite{Triangle-EnclArea-JA86}.
The minimum perimeter triangle, minimum area rectangle, and minimum area parallelogram, enclosing a convex polygon, can also be found in $O(n)$ time \cite{Triangle-EnclPeri-JCDCG02,Classic-rcaliper-MEL83,3dorEncl-areaParallelogram-SOCG95}.
The minimum perimeter enclosing parallelogram is partially solved (using an unproved conjecture) in \cite{3dorEncl-periParallelogram-cccg10} .

The \emph{rotating-caliper technique} \cite{Classic-rcaliper-MEL83} can be applied in many variants of the enclosing problems, including the case of triangle, rectangle, and parallelogram.
    For different variants of the inclusion problems, however, it seems different techniques must be employed.
Usually, the inclusion problems are more difficult than their enclosing counterparts \cite{Othershape-square-Allerton87}.
    So it is not a shame that our algorithm is slower than its corresponding one in \cite{3dorEncl-areaParallelogram-SOCG95}; the latter, though fast, is essentially easier.

\medskip The extremal ellipsoid problems ask the maximum enclosed ellipse (or ellipsoid in $\mathbb{R}^d$ space) of $P$ (defined by $n$ half-plane boundaries) and
          the minimum enclosing ellipse (or ellipsoid in $\mathbb{R}^d$) of $P$ (defined by $n$ vertices).
        They are \emph{LP-type problems} and can be solved in $O(n)$ time for fixed $d$ \cite{ellipse-LPtype-JA96,ellipse-LPtype-SCG92}.
          Alternatively, they can be formulated as \emph{convex programming problems} and thus be solved in $O(n)$ time \cite{ellipse-spanconvex-SCG92,ellipse-ipconvex-tr,ellipse-book-co}.

    Is our problem also a convex programming problem?
  As the MAE is unique (\cite{ellipse-john-CAV48,ellipse-book-co}),
    whereas there could be multiple LMAPs (for instance, there are five LMAPs if $P$ is a regular pentagon; see \ref{sect:many-pentagon}),
      which are locally optimal solutions,
      the answer is probably no.
    In \ref{sect:address-cp}, we show a quadratic programming formulation (\ref{eqn:plausibleCP}) of the MAP problem
      that looks like but is \textbf{not} a convex programming.

\subsection{More related work (maximum parallelepiped in convex bodies)}

The maximum volume parallelepiped in convex bodies has been studied extensively in convex geometry.
Assume $C$ is a convex body in $\mathbb{R}^d$ and $Q$ is the maximum volume parallelepiped in $C$.
\cite{Math-approx-GD98} proved that the concentric scaling of $Q$ by a factor of $2d-1$ covers $C$; and
\cite{Math-approx-JDG04} proved that there exists one scaling of $Q$ by a factor of $d$ which covers $C$.
A closely related research is the maximum volume ellipsoid (MVE) in convex bodies.
In his seminal paper \cite{ellipse-john-CAV48}, Fritz John proved that inside every convex body there is a unique MVE, and the concentric expansion of the MVE by a factor of $d$ contains the convex body.
\cite{Math-3dinscribed-french54} proved that any convex body in $\mathbb{R}^3$ admits an inscribed parallelepiped.
This is in general not true in $\mathbb{R}^d$ for $d\geq 5$; see \cite{Math-survey-GM97}.
For any planar convex closed curve $C$, there is a parallelogram inscribed in $C$ whose area is at least $\frac{1}{2}$ times the area surrounded by $C$ \cite{Math-PinC-AMM60}.
A similar result in $\mathbb{R}^3$ is proved in \cite{Math-3dinscribed-french52}, and the ratio is $3!/3^3$; see also \cite{Math-PinS-DM00,Math-PinS-DCG99}.
See \cite{Math-PinE-AMM07,Math-SafeDomain-EJP04} for other interesting results.

\subsection{Motivation from the Heilbronn triangle problem}\label{subsect:heil}

The Heilbronn triangle problem (\cite{Heilbronn-convexhull-LMS71}) is a classic minimax problem in discrete geometry,
    which concerns placing $m$ points in a convex region, in order to avoid small triangles spanned by these points.
Polynomial algorithms were given for finding considerably good placements (\cite{Heilbronn-Lefmann-JC00}, \cite{Heilbronn-Lefmann3d-JC02}).
On finding the optimal placement, the following lemma states that the simplest nontrivial case, namely $m=4$, reduces to finding the MAP and the maximum area triangle (MAT) in the region.
As a corollary, if the given region is a convex polygon with $n$ vertices, the case $m=4$ can be solved in $O(n^2)$ time by combining our algorithm with the algorithm for the MAT in \cite{Triangle-ultimate-Arxiv}.

\begin{lemma}[Reduction]
  Solving the Heilbronn triangle problem in a convex region $K$ for $m=4$ reduces to finding the MAP and the MAT in $K$.
\end{lemma}

\begin{proof}
For any points $A,B,C,D$, denote by $h(A,B,C,D)$ the area of the smallest triangle among $\triangle ABC,\triangle ABD,\triangle ACD$ and $\triangle BCD$.
When $m=4$, the Heilbronn triangle problem asks the optimal location of the points $(A,B,C,D)$ in $K$, so that $h(A,B,C,D)$ is maximum.
Let $A_1B_1C_1$ be an MAT in $K$, whose area equals $t$.
  Let $A_2B_2C_2D_2$ be an MAP in $K$, whose area equals $p$.
  We state three observations.
\begin{enumerate}
\item \emph{There is a location of $(A,B,C,D)$ such that $h(A,B,C,D)=t/3$.}\\
    Proof: Let $(A,B,C)=(A_1,B_1,C_1)$ and $D$ be the centroid of $\triangle A_1B_1C_1$.\smallskip
\item \emph{There is a location of $(A,B,C,D)$ such that $h(A,B,C,D)=p/2$.}\\
    Proof: Let $(A,B,C,D)=(A_2,B_2,C_2,D_2)$. \smallskip
\item \emph{For every location of $(A,B,C,D)$, we have $h(A,B,C,D)\leq \max (t/3,p/2)$.}
\end{enumerate}

\begin{figure}[h]
    \centering \includegraphics[width=.5\textwidth]{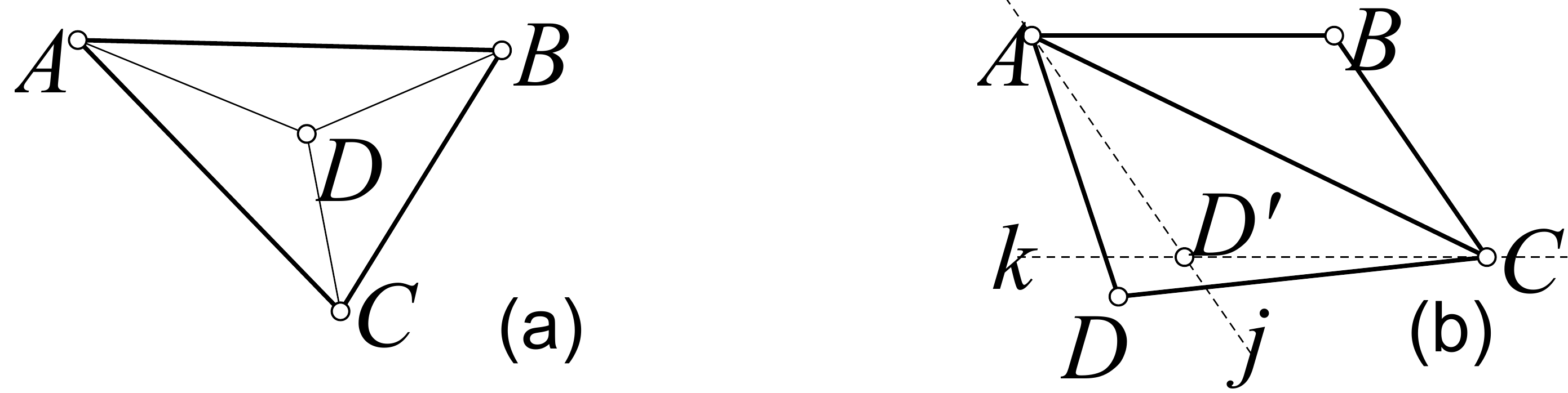}
    \caption{Illustration of Lemma Reduction.}\label{fig:reduction}
\end{figure}

To prove the last observation, we discuss two cases:
\begin{enumerate}
\item[Case~1] the convex hull of $A,B,C,D$ is a triangle. See Figure~\ref{fig:reduction}~(a).
    Without loss of generality, assume $D\in \triangle ABC$.
   Then, $h(A,B,C,D)=\min\{Area(ABD), Area(BCD)$, $Area(CAD)\}\leq Area(ABC)/3 \leq t/3$.
\item[Case~2] the convex hull of $A,B,C,D$ is a quadrilateral. See Figure~\ref{fig:reduction}~(b).
    Without loss of generality, assume $h(A,B,C,D)=\triangle ABC$.
  Draw a line $j$ at $A$ that is parallel to $BC$. Draw a line $k$ at $C$ that is parallel to $AB$.
    Denote by $D'$ the intersection of $j$ and $k$.  We claim that  $D'\in \triangle ACD$.
        The proof is as follows. First, observe that $D,B$ do not lie on the same side of $j$, since otherwise $BCD$ would have a smaller area than $ABC$.
	Similarly, $D,A$ do not lie on the same side of $k$. These together imply that $D'\in \triangle ACD$.
    Therefore, $ABCD'$ is a parallelogram that lies in $K$.
    Therefore, $h(A,B,C,D)=Area(\triangle ABC)=Area(ABCD')/2\leq p/2$.
\end{enumerate}

Altogether, $\max\{h(A,B,C,D)\}=\max(t/3,p/2)$. Moreover, by the above analysis, computing an optimal location reduces to finding the MAT and MAP in $K$.
\end{proof}

\section{Preliminaries and Technique Overview}\label{sect:preliminiary}

Denote the boundary of $P$ by $\partial P$.
Let $e_1,\ldots,e_n$ be a clockwise enumeration of the edges of $P$.
Let $v_1,\ldots,v_n$ be the vertices of $P$ such that $e_i=(v_i,v_{i+1})$ where $v_{n+1}=v_1$.
Let $|AB|$ denote the distance between any pair of points $A$ and $B$.

\medskip \noindent\textbf{Key assumptions.}
We regard $P$ as a \textbf{compact} set; namely, it contains its boundary and interior.
If a point is said lying in $P$, it may lie in $P$'s boundary.
We regard all edges of $P$ as \textbf{open} segments; namely, they do not contain their endpoints.
If a point is said lying in $e_i$, it does not lie on any endpoint of $e_i$.
For simplicity of discussion, we assume that all edges of $P$ are \textbf{pairwise-nonparallel}.
Unless otherwise stated, an edge or a vertex refers to an edge or a vertex of $P$, respectively.

\begin{definition}\label{def:locallymaximal}
We say a parallelogram lies in $P$ if all its corners lie in $P$.
Consider any parallelogram $Q=A_0A_1A_2A_3$ that lies in $P$.
We say $Q$ is \emph{maximum}, if it has the largest area among all parallelograms that lie in $P$.
We say $Q$ is \emph{locally maximal}, if it has an area larger than or equal to all its sufficiently close parallelograms that lie in $P$;
formally, if $\exists \delta>0 \text{ such that }\forall Q'\in N_\delta(Q), Area(Q)\geq Area(Q'), \text{ where}$
\[N_\delta(A_0A_1A_2A_3)=\{B_0B_1B_2B_3\text{ is a parallelogram in }P\mid \forall i, |A_iB_i|<\delta\}.\]
\end{definition}

\begin{definition}[inscribed \& slidable]
A parallelogram is \emph{inscribed}, if all its corners lie in $P$'s boundary.
A parallelogram is \emph{slidable}, if it has two corners lying in the same edge of $P$.
(If corner $A$ lies in $e_i$ whereas corner $B$ lies on some endpoint of $e_i$,
  these two corners are \textbf{not} counted as lying in the same edge, since the endpoint of $e_i$ does not belong to $e_i$.)
A parallelogram is \emph{non-slidable} if it is not slidable.
\end{definition}

\begin{lemma}\label{lemma:local-maximal-non-slidable}
1. A locally maximal parallelogram must be inscribed.
2. For any inscribed slidable parallelogram, there is an inscribed non-slidable parallelogram with the same area.
(The proof of this lemma is trivial and deferred to \ref{sect:omit-slidable}.)
\end{lemma}

\begin{definition}[MAP \& LMAP]\label{def:LMAP}
A parallelogram is an \emph{MAP} (Maximum Area Parallelogram) if it is maximum and non-slidable,
and is an \emph{LMAP} (Locally Maximal Area Parallelogram) if it is locally maximal and non-slidable
(an MAP is always an LMAP).
We safely exclude the slidable parallelograms due to Lemma~\ref{lemma:local-maximal-non-slidable}.
We have to exclude them since otherwise there could be infinitely many LMAPs and MAPs.
\end{definition}

\noindent\textbf{Unit.} A \emph{unit} of $P$ refers to an edge or a vertex of $P$. There are $2n$ units.\smallskip

\noindent\textbf{Boundary-portion.} A \emph{boundary-portion} of $P$ refers to a continuous portion of $\partial P$.\smallskip

\newcommand{\dist}{\mathsf{disprod}}
\newcommand{\el}{\ell}        
\noindent\textbf{Distance-product.} The \emph{distance-product} from a point $X$ to two lines $l,l'$,
  denoted by $\dist_{l,l'}(X)$, is defined to be $d_l(X)\cdot d_{l'}(X)$, where $d_l(X)$ denotes the distance from $X$ to $l$.
In this paper, we mainly focus on the case where $l,l'$ are extended lines of the edges of $P$.
For convenience, denote the extended line of $e_i$ by $\el_i$.

\subsection{Technique overview \& Organization of the paper}\label{subsect:techover}

By Lemma~\ref{lemma:local-maximal-non-slidable}, the LMAPs must be inscribed in $P$.
However, there are infinite many inscribed parallelograms:
  Given two distinct points $A,A'$ in $\partial P$, if we find the other chord of $P$ that is a translate of chord $\overline{AA'}$, we obtain an inscribed parallelogram. In order to find the LMAPs, we must prove more constraints of the LMAPs.

\smallskip \noindent \textbf{A group of constraints (preview).}\mbox{  }
  We find out $\Theta(n^2)$ boundary-portions $\{\zeta(u,u')\}\mid u,u'\text{ are two units}\}$ so that the following hold:
  For every corner $A_i$ of every LMAP $Q=A_0A_1A_2A_3$, if $A_{i+1},A_{i-1}$ lie in units $u,u'$ respectively
    (throughout, assume that $A_0,A_1,A_2,A_3$ lie in clockwise order, and subscripts of $A$ taken modulo 4), then $A_i$ is restricted to $\zeta(u,u')$.
    In other words, we find a range $\zeta(u,u')$ for bounding $A_i$ given the units $u$ and $u'$ containing $A_i$'s neighboring corners.

How do we define $\zeta(u,u')$ is crucial and is explained in the following.

\medskip \noindent \textbf{On defining $\zeta(u,u')$.}\mbox{  } Some observations are required for this definition.
Section~\ref{sect:Z} proves that for every pair $(\el_i,\el_j)$, in domain $P$, function $\dist_{\el_i,\el_j}()$ achieves maximum value at a unique point (denoted by $Z_i^j$) and (i) this function is unimodal in the boundary-portion connecting $e_i,e_j$ and containing $Z_i^j$.
  Section~\ref{sect:identity} proves that given two nonparallel lines $l,l'$ and two points $X,X'$, only one parallelogram has a pair of opposite corners lying on $l,l'$ and the other
    two corners lying at $X,X'$, and more importantly, (ii) its area is proportional to $|\dist_{l,l'}(X)-\dist_{l,l'}(X')|$.

\smallskip Combining these observations, we define (deduce) $\zeta(u,u')$ in Section~\ref{sect:clamping} using the following idea.
    Assume parallelogram $Q=A_0A_1A_2A_3$ is inscribed and $A_{i+1},A_{i-1}$ lie in units $u,u'$ respectively.
        We aim to find some (bad) (as wide as possible) boundary-portion(s) $\alpha$ between $u'$ and $u$ (clockwise), so that $Q$ cannot be an LMAP as long as $A_i\in \alpha$.
      If so, by the contrapositive, we can get $A_i\notin \alpha$ when $Q$ is an LMAP, and hence we can define the complementary portion of $\alpha$ to be $\zeta(u,u')$.

    The challenge lies in arguing that $Q$ is not an LMAP.
      Choose $l,l'$ from $\el_1,\ldots,\el_n$ which respectively contain $u,u'$ (the choice is not unique if $u$ or $u'$ is a vertex).
    We can construct a parallelogram $Q'$ from $Q$ (and sufficiently close to $Q$) by fixing $A_{i+2}$ and (slightly) \emph{adjusting $A_i$ along $\partial P$},
       while restricting the other two corners on $l,l'$ respectively.
    Applying (ii), $Area(Q')-Area(Q)$ is proportional to the change of $\dist_{l,l'}(A_{i})$.
    Then, by analyzing the change $\dist_{l,l'}(A_{i})$ with respect to the movement of $A_i$ as depicted in (i),
      we might be able to deduce that $Area(Q')>Area(Q)$ and thus conclude that $Q$ is not locally maximal and hence not an LMAP.

\newcommand{\block}{\mathsf{block}}
\medskip\noindent\textbf{Algorithm.}\mbox{  }
We try every pair of units $(u,u')$ as the units containing two opposite corners of an LMAP,
  and try every vertex lying in $\zeta(u,u')$ as the third corner (the case where the third corner is not on a vertex of $P$ is handled by a symmetric algorithm) and then compute the last corner.
   The running time is bounded by $O(n^2)$ using monotonicity properties of $\{\zeta(u,u')\}$.
   The algorithms are given in section~\ref{sect:algo}.

\subsection{More basic notions and frequently applied notations}

\begin{definition}[\textbf{``Chasing''}]
Edge $e_i$ is \emph{chasing} $e_j$, denoted by $e_i\prec e_j$,
  if the intersection of $\el_i$ and $\el_j$ lies between $e_i$ and $e_j$ clockwise.
For example, in Figure~\ref{fig:chasing}, $e_1$ is chasing $e_2$ and $e_3$, whereas $e_4,\ldots,e_7$ are chasing $e_1$.
By pairwise-nonparallel assumption of edges, for any pair of edges, exactly one of them is chasing the other.
\end{definition}

\begin{figure}[b]
\begin{minipage}[b]{.45\textwidth}
\centering\includegraphics[width=.32\textwidth]{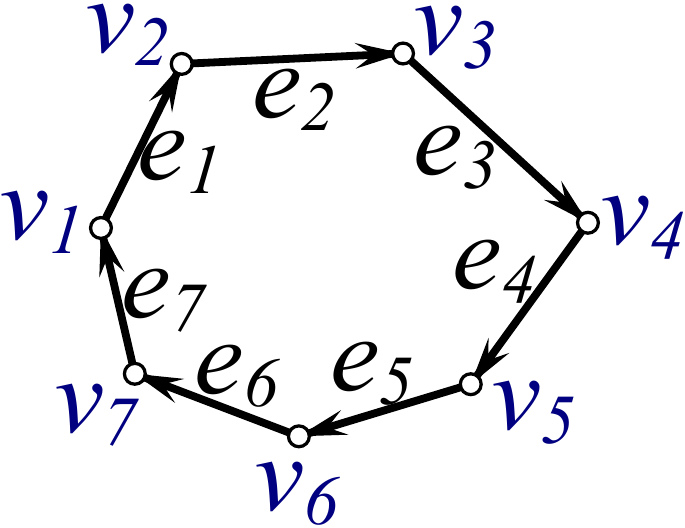}
\caption{Chasing relation}\label{fig:chasing}
\end{minipage}
\begin{minipage}[b]{.45\textwidth}
\centering\includegraphics[width=.35\textwidth]{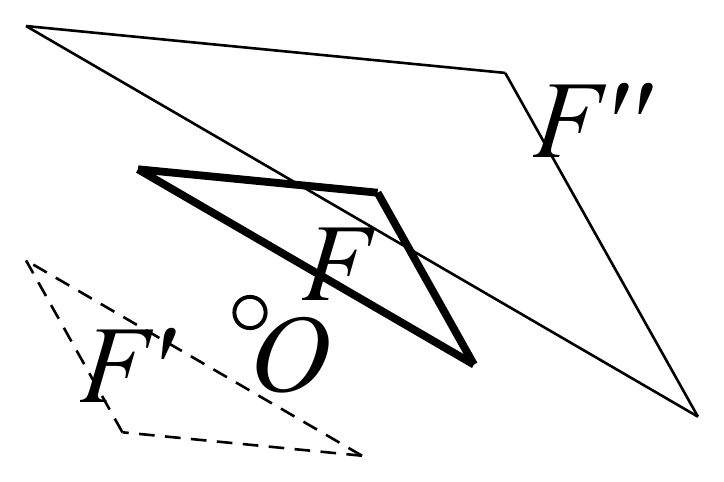}
\caption{Reflect \& Scale}\label{fig:scale-reflect}
\end{minipage}
\end{figure}

The \textbf{\emph{reflection}} and \textbf{\emph{scaling}} of a figure $F$ with respect to a point $O$ is defined in a standard manner:
The \emph{reflection} is the figure $F'$ which is congruent to $F$ and is centrally-symmetric to $F$ with respect to $O$,
and the \emph{$k$-scaling} is the figure $F''$ which contains point $X$ if and only if $F$ contains $(X-O)/k+O$. See Figure~\ref{fig:scale-reflect}.

We consider every boundary-portion \emph{directed} and its direction conforms with the \emph{clockwise order} of $\partial P$.
Its two endpoints are referred to as its \emph{starting and terminal points} in the way that conforms with the clockwise order.
Given points $X,X'$ on $\partial P$, we will pass through a boundary-portion if we travel along $\partial P$ in clockwise from $X$ to $X'$;
  its endpoints-inclusive and endpoints-exclusive versions are denoted by $[X\circlearrowright X']$ and $(X\circlearrowright X')$.
Note that $[X\circlearrowright X']$ only contains the single point $X$ when $X=X'$.
Given points $A,B$ on a boundary-portion $\rho$, we state $A<_\rho B$ if $A$ appears earlier than $B$ traveling along $\rho$. Denote by $A\leq_\rho B$ if $A=B$ or $A<_\rho B$.

\newcommand{\M}{\mathsf{M}}   
\newcommand{\D}{\mathsf{D}}
\newcommand{\I}{\mathsf{I}}
Denote by $e_i\preceq e_j$ if $e_i=e_j$ or $e_i\prec e_j$.
Let $\M(A,B)$ denote the mid point of points $A$ and $B$.
For two edges $e_i$ and $e_j$, we denote by $\I_{i,j}$ the intersection of $\el_i$ and $\el_j$.
Denote by $\D_i$ the unique vertex with largest distance to $\el_i$.
    The uniqueness follows from the pairwise-nonparallel assumption of edges.

\section{The distance-product function and the $Z$-points}\label{sect:Z}

In this section, we show that function $\dist_{l,l'}(X)$ enjoys many interesting properties, especially when we select two extended lines $\el_i,\el_j$ of the edges of $P$ as $l,l'$.

\begin{lemma}[\textbf{Strict concavity of $\dist_{l,l'}$}]\label{lemma:dist_concave}
\begin{enumerate}
\item Given nonparallel lines $l,l'$ in the plane. Assume points $B,B'$ lie on $l,l'$, respectively, and neither of them lie on the intersection of $l,l'$.
    Then, $\dist_{l,l'}()$ is strictly concave on $\overline{BB'}$ and maximized at $\M(B,B')$.
\item Assume moreover there are two \textbf{distinct} points $X,X'$ in $\overline{BB'}$ such that $B,X,X',B'$ lie in this order.
Then, the following hold.\\
(a) We have $|BX'|\leq \frac{1}{2}|BB'|$ if $\dist_{l,l'}()$ is maximized at $X'$ in $\overline{XX'}$.\\
(b) The signs of $|BX|-|B'X'|$ and $\dist_{l,l'}(X)-\dist_{l,l'}(X')$ are the same.
\end{enumerate}
\end{lemma}

\begin{proof}
Suppose $X$ is a point on segment $\overline{BB'}$ and its distance to $B$ is $x$, as shown in Figure~\ref{fig:Z_basic}~(a).
Obviously, we have the following formula:
\[\dist_{l,l'}(X)= x \sin \angle B \cdot (|BB'|-x) \sin \angle B'= k\cdot x (|BB'|-x),\]
where $k$ is a constant.
Therefore, by calculus, it is strictly concave on $\overline{BB'}$ and maximized at $x=\frac{1}{2}|BB'|$.
Part~2 also easily follows from the formula.
\end{proof}

\begin{figure}[h]
  \centering \includegraphics[width=.85\textwidth]{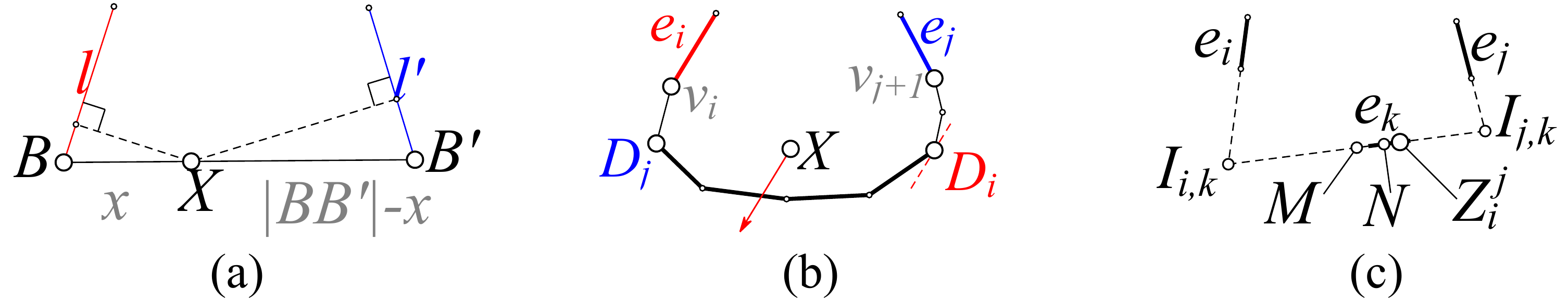}
  \caption{Illustrations of the proofs of Lemmas~\ref{lemma:dist_concave}  and \ref{lemma:dist-unique-location}}\label{fig:Z_basic}
\end{figure}

\begin{lemma}\label{lemma:dist-unique-location}
Given two edges $e_i,e_j$ such that $e_i\prec e_j$.
\begin{enumerate}
\item In domain $P$, function $\dist_{\el_i,\el_j}()$ achieves maximum value at a unique point;
    denoted as $Z_{e_i}^{e_j}$ or $Z_i^j$ henceforth.
        (The points in $\{Z_i^j\mid e_i\prec e_j\}$ are called \textbf{\emph{$Z$-points}}. All of them are lying in $\partial P$ due to the second part of this lemma).
\item Point $Z_i^j$ lies in $\partial P$. Moreover, it lies in both $[\D_i\circlearrowright\D_j]$ and $(v_{j+1}\circlearrowright v_i)$.
\item If $Z_i^j$ lies in some edge $e_k$, it lies at the mid point of $\I_{i,k}$ and $\I_{j,k}$.
\end{enumerate}
\end{lemma}

\begin{proof}
1. If $\dist_{\el_i,\el_j}()$ achieves maximum value at two distinct points in $P$, e.g.\ $X_1$ and $X_2$,
applying the concavity of $\dist_{\el_i,\el_j}$ (Lemma~\ref{lemma:dist_concave}),
  the points between $X_1$ and $X_2$ have larger distance-products to $(\el_i,\el_j)$ than $X_1$ and $X_2$; contradictory.

\smallskip\noindent
2. We state that (i) $Z_i^j\in [\D_i\circlearrowright v_i]$; and (ii) $Z_i^j\in [v_{j+1}\circlearrowright\D_j]$.
Combining (i) and (ii), we get $Z_i^j \in [\D_i\circlearrowright \D_j]$.
Moreover, because $[\D_i\circlearrowright\D_j]\subseteq[v_{j+1}\circlearrowright v_i]$ while $Z_i^j$ obviously cannot lie on $v_{j+1}$ or $v_i$, point $Z_i^j$ must lie in $(v_{j+1}\circlearrowright v_i)$.

The proof of (i) is as follows; proof of (ii) is symmetric and omitted.
Take any point $X$ that lies in $P$ but not in $[\D_i\circlearrowright v_i]$, we shall prove that $Z_i^j\neq X$.
See Figure~\ref{fig:Z_basic}~(b). Make a ray at $X$ which has the opposite direction to $e_i$.
By the assumption of $X$, there is another point $X'$ on the ray which still lies in $P$.
Clearly, $X'$ has a larger distance-product to $(\el_i,\el_j)$ than $X$, which implies that $Z_i^j\neq X$.

\smallskip \noindent
3. Suppose $Z_i^j\neq \M(\I_{i,k},\I_{j,k})$ as shown in Figure~\ref{fig:Z_basic}~(c). (In this figure, $M$ denotes $\M(\I_{i,k},\I_{j,k})$.)
There must be a point $N$, which lies in $e_k$ and lies between $Z_i^j$ and $\M(\I_{i,k},\I_{j,k})$. Applying the strict concavity of $\dist_{\el_i,\el_j}()$ on $\overline{\I_{i,k}\I_{j,k}}$ (Lemma~\ref{lemma:dist_concave}), $\dist_{\el_i,\el_j}(N)>\dist_{\el_i,\el_j}(Z_i^j)$, which contradicts the definition of $Z_i^j$.
 \end{proof}

\begin{lemma}[\textbf{Unimodality of $\dist_{\ell_i,\ell_j}$}]\label{lemma:dist_unimodal}
For edge pair $(e_i,e_j)$ such that $e_i\prec e_j$, $\dist_{\ell_i,\ell_j}()$ is \emph{strictly unimodal} on $[v_{j+1}\circlearrowright v_i]$. Specifically, $\dist_{\ell_i,\ell_j}(X)$\\
\quad (1) strictly increases when $X$ travels from $v_{j+1}$ to $Z_i^j$ in clockwise along $\partial P$, and\\
\quad (2) strictly decreases when $X$ travels from $Z_i^j$ to $v_i$ in clockwise along $\partial P$.
\end{lemma}

\begin{proof}
We prove (2); the proof of (1) is symmetric.

First, consider the traveling process of $X$ from $Z_i^j$ to $v_k$, where $v_k$ denotes the clockwise first vertex in $[Z_i^j\circlearrowright \D_j]$ that is not equal to $Z_i^j$.
See Figure~\ref{fig:Z_unimodal}~(a). Let $A=\I_{k-1,i},B=\I_{k-1,j}$.
By definition, $Z_i^j$'s distance-product to $(\el_i,\el_j)$ is superior to all the other points on $\overline{v_kZ_i^j}$,
  which implies that $|Av_k|<|AZ_i^j|\leq \frac{1}{2}|AB|$ by Lemma~\ref{lemma:dist_concave}.
Again by Lemma~\ref{lemma:dist_concave}, this inequality implies that when $X$ travels from $Z_i^j$ to $v_k$, its distance-product to $(\el_i,\el_j)$ strictly decreases.

Next, consider the travel of $X$ from $v_k$ to $v_{k+1}$. See Figure~\ref{fig:Z_unimodal}~(b). Let $A'=\I_{k,i}$, $B'=\I_{k,j}$.
Make a line at $A'$ parallel to $e_j$ and assume it intersects $\el_{k-1}$ at point $C$.
Because $A'C$ is parallel to $BB'$, we get $|A'v_k|:|B'v_k|=|Cv_k|:|Bv_k|<|Av_k|:|Bv_k|$.
Because $|Av_k|<\frac{1}{2}|AB|$, we get $|Av_k|<|Bv_k|$.
Together, $|A'v_k|<|B'v_k|$. So, $|A'v_k|<\frac{1}{2}|A'B'|$.
Therefore, $|A'v_{k+1}|<|A'v_k|<\frac{1}{2}|A'B'|$.
This means $\dist_{\el_i,\el_j}(X)$ strictly decreases when $X$ goes from $v_k$ to $v_{k+1}$ due to Lemma~\ref{lemma:dist_concave}.

\smallskip By induction, $\dist_{\el_i,\el_j}(X)$ strictly decreases before $X$ arrives at $\D_j$.
Finally, consider the traveling process from $\D_j$ to $v_i$.
In this process, $\dist_{\ell_i,\ell_j}(X)$ strictly decreases because both $d_{\el_i}(X)$ and $d_{\el_j}(X)$ strictly decrease.
 \end{proof}

\begin{figure}[h]
\begin{minipage}{.66\textwidth}
\centering\includegraphics[width=.85\textwidth]{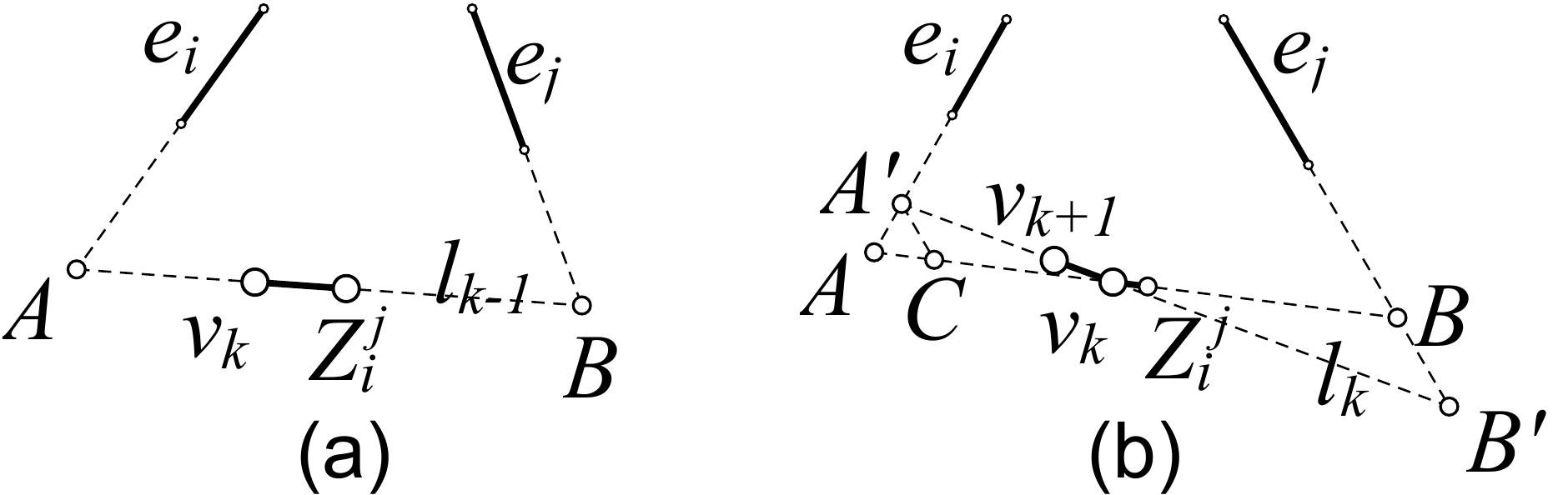}
\caption{Illustration of the proof of Lemma~\ref{lemma:dist_unimodal}}\label{fig:Z_unimodal}
\end{minipage}
\begin{minipage}{.33\textwidth}
\centering\includegraphics[width=.73\textwidth]{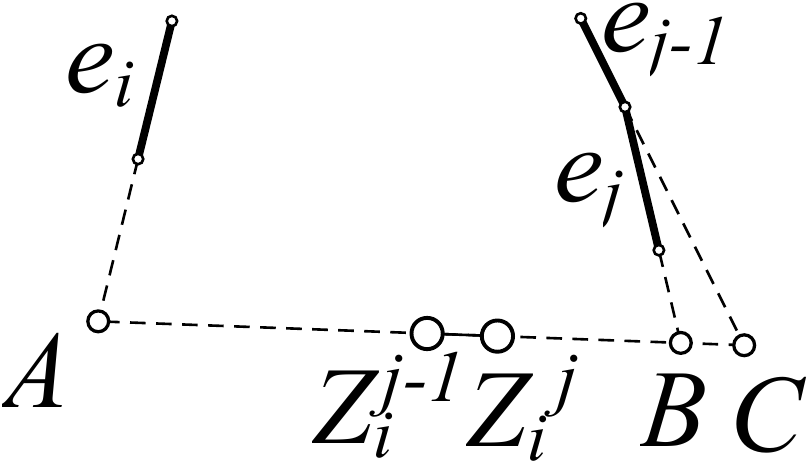}
\caption{Illustration of the proof of Lemma~\ref{lemma:Z_bi-monotonicity}}\label{fig:Z_monotone}
\end{minipage}
\end{figure}

\begin{lemma}[\textbf{Bi-monotonicity of $Z$-points}]\label{lemma:Z_bi-monotonicity}
Given $e_s,e_t$ such that $e_s\preceq e_t$. Let $S=\{(e_i,e_j)\mid e_i\prec e_j, \text{and $e_i,e_j$ both belong to $\{e_s,e_{s+1},\ldots,e_t\}$.}\}$
We claim that all the $Z$-points in set $\{Z_i^j\mid (e_i,e_j)\in S\}$ lie in $\rho=[v_{t+1} \circlearrowright v_s]$ and they obey the following \textbf{bi-monotonicity}:
For $(e_{i},e_{j})\in S$ and $(e_{i'},e_{j'})\in S$,
$$\text{if $e_{i}\preceq e_{i'}$ and $e_{j}\preceq e_{j'}$, then $Z_{i}^{j}\leq_\rho Z_{i'}^{j'}$}.$$
\end{lemma}

\begin{proof}
Assume that $e_s\prec e_t$, otherwise $e_s=e_t$ and the claim is trivial.

Assume that $(e_i,e_j)\in S$. According to Lemma~\ref{lemma:dist-unique-location}, $Z_i^j$ lies in $[\D_i\circlearrowright \D_j]$.
Since $e_s\prec e_t$, we have $[\D_i\circlearrowright \D_j]\subseteq [v_{t+1} \circlearrowright v_s]$.
Together, $Z_i^j$ lies in $\rho=[v_{t+1} \circlearrowright v_s]$.

Proving the monotonicity of the $Z$-points reduces to proving the following facts:
\emph{If $(e_i,e_j)$ belongs to $S$ and $e_i,e_j$ are not adjacent, then} $Z_i^{j-1}\leq_\rho Z_i^j\text{ and }Z_i^j \leq_\rho Z_{i+1}^j.$

We prove the first inequality; the other is symmetric.
See Figure~\ref{fig:Z_monotone}. Suppose to the contrary that $Z_i^j<_\rho Z_i^{j-1}$.
  The line connecting these two $Z$-points intersects with $\el_i,\el_j,\el_{j-1}$, and we denote the intersections by $A,B,C$, respectively.
  Applying the concavity of $\dist_{\el_i,\el_j}()$ on segment $\overline{AB}$ (see Lemma~\ref{lemma:dist_concave}), we get $|AZ_i^j|\leq \frac{1}{2}|AB|$.
  Applying the concavity of $\dist_{\el_i,\el_{j-1}}()$ on segment $\overline{AC}$, we get $|AZ_i^{j-1}|\geq \frac{1}{2}|AC|$.
  Together, we get $|AC|<|AB|$. This contradicts with the assumption of $A,B,C$.
 \end{proof}

\begin{lemma}[Computational aspect of the $Z$-points]\label{lemma:Z-compute}~
\begin{enumerate}
\item Given $e_i,e_j$ such that $e_i\prec e_j$,
    the position of $Z_i^j$ can be computed in $O(1)$ time if the unit containing $Z_i^j$ is known. (Recall unit = vertex or edge.)
\item Assume $e_i\prec e_j$ and $v_k$ lies in $(v_{j+1} \circlearrowright v_i)$.
    By Lemma~\ref{lemma:dist-unique-location}, $Z_i^j\in (v_{j+1} \circlearrowright v_i)$.
    So, the position of $Z_i^j$ has the three possibilities:
    (i) $=v_k$; (ii) $\in (v_{j+1} \circlearrowright v_k)$; or (iii) $\in (v_k \circlearrowright v_i)$.
    Given $i,j,k$, we can distinguish these cases in $O(1)$ time.
\item Given $m$ pairs of edges $(a_1,b_1),\ldots,(a_m,b_m)$ so that $a_i\prec b_i$ for $1\leq i\leq m$, and that $a_1,\ldots,a_m$ lie in clockwise order around $\partial P$ and $b_1,\ldots,b_m$ lie in clockwise order around $\partial P$, we can compute $Z_{a_1}^{b_1},\ldots,Z_{a_m}^{b_m}$ all together in $O(m+n)$ time.
\end{enumerate}
\end{lemma}

\begin{proof}
1. If the unit containing $Z_i^j$ is a vertex, $Z_i^j$ can be computed directly;
otherwise, $Z_i^j$ can be computed in $O(1)$ time according to part~3 of Lemma~\ref{lemma:dist-unique-location}.

\smallskip\noindent
2. We say that a point $X$ \emph{dominates} point $X'$, if $\dist_{\el_i,\el_j}(X)>\dist_{\el_i,\el_j}(X').$
The unimodality of $\dist_{\el_i,\el_j}$ (Lemma~\ref{lemma:dist_unimodal}) implies the following facts:\\
1. ``$Z_i^j$ lies on $v_k$'' if and only if ``$v_k$ dominates all points on $e_{k-1}$ and $e_k$.''\\
2. ``$Z_i^j$ lies in $(v_{j+1}\circlearrowright v_k)$'' if and only if ``some point in $e_{k-1}$ dominates $v_k$.''\\
3. ``$Z_i^j$ lies in $(v_k \circlearrowright v_i)$'' if and only if ``some point in $e_k$ dominates $v_k$.''

Thus, it reduces to answering the following queries:

\quad \emph{Does $v_k$ dominate every point in $e_{k-1}$?} \emph{Does $v_k$ dominate every point in $e_k$?}

We can answer them in $O(1)$ time by applying the concavity of $\dist_{\el_i,\el_j}$.

\smallskip\noindent
3. By part~1, to compute $Z_{a_1}^{b_1},\ldots,Z_{a_m}^{b_m}$, we only need to determine the respective units that they lie on.
By the bi-monotonicity of the $Z$-points, $Z_{a_1}^{b_1},\ldots,Z_{a_m}^{b_m}$ lie in clockwise order,
so the units they lie on are also in clockwise order.
So, we can walk around $\partial P$ to compute these $Z$-points in order, which costs $O(m+n)$ time.
 \end{proof}

\section{Calculating area of parallelogram using distance-product}\label{sect:identity}

In this section, we show that if a parallelogram has a pair of opposite corners restricted to lines $l,l'$ respectively,
  its area is proportional to $|\dist_{l,l'}(X)-\dist_{l,l'}(X')|$, where $X,X'$ are the positions of the other two corners.
  This simple connection between the area and the distance-product is crucial to study the LMAPs.

Let $r_O(F)$ denote the reflection of figure $F$ with respect to point $O$.
\begin{claim}\label{claim:ll'}
Given nonparallel lines $l,l'$ and a point $M$ in the plane.
Let $Y$ be the intersection of $l$ and $r_M(l')$, and $Y'$ the intersection of $l'$ and $r_M(l)$.
Then, $(Y,Y')$ is the unique pair of points such that $Y\in l,Y'\in l'$ and $\M(Y,Y')=M$.
\end{claim}

\begin{proof}
Since $Y$ is the intersection of $l$ and $r_M(l')$,
  its reflection $r_M(Y)$ equals to the intersection of $r_M(l)$ and $r_M(r_M(l'))=l'$.
  So, $r_M(Y)=Y'$, i.e.\ $\M(Y,Y')=M$.

We now prove the uniqueness. Assume $(B,B')$ satisfy that $B\in l,B'\in l'$ and $\M(B,B')=M$.
    Because $\M(B,B')=M$, we know $B=r_M(B')$. This implies $B\in r_M(l')$ because $B'\in l'$.
    Further since $B\in l$, we get $B=Y$. Similarly, $B'=Y'$.
 \end{proof}

\newcommand\parallelogram{\mathord{\text{\tikz[baseline] \draw (0,.1ex) -- (.8em,.1ex) -- (1em,1.6ex) -- (.2em,1.6ex) -- cycle;}}}

Henceforth in this section, we assume that $l,l'$ are two lines intersecting at $O$ and $X,X'$ are two points which lie in the some quadrant divided by $l,l'$.

A parallelogram is \emph{degenerate} if all of its four corners lie in the same line.

\begin{claim}[A corollary of Claim~\ref{claim:ll'}]\label{claim:XX'll'}
See Figure~\ref{fig:def-XXLL}. Let $M=\M(X,X')$.
Let $Y$ be the intersection of $l$ and $r_M(l')$, and $Y'$ be the intersection of $l'$ and $r_M(l)$.
Using Claim~\ref{claim:ll'}, $\M(Y,Y')=M=\M(X,X')$. So, $XYX'Y'$ is a parallelogram (which may be degenerate)
and we denote it by $\parallelogram(X,X',l,l')$.
We claim that $\parallelogram(X,X',l,l')$ is the unique parallelogram (which may be degenerate) which has a pair of opposite corners lying on $X,X'$ and has two other corners on $l,l'$ respectively.
\end{claim}

\begin{lemma}\label{lemma:area-pre}
By comparing $\dist_{l,l'}(X)$ with $\dist_{l,l'}(X')$, we can infer the following relations between the four corners $X,X',Y,Y'$ of $\parallelogram(X,X',l,l')$.
\begin{enumerate}
  \item If $\dist_{l,l'}(X)<\dist_{l,l'}(X')$, then $X\in \triangle OYY'$ and $X'\notin \triangle OYY'$.
  \item If $\dist_{l,l'}(X)>\dist_{l,l'}(X')$, then $X\notin \triangle OYY'$ and $X'\in \triangle OYY'$.
  \item If $\dist_{l,l'}(X)=\dist_{l,l'}(X')$, then $\parallelogram(X,X',l,l')$ is degenerate.
\end{enumerate}
\textbf{Note:} Here $\triangle OYY'$ is regarded as a closed set, so it contains its boundary.\smallskip
\end{lemma}

\begin{figure}[h]
\begin{minipage}{0.26\textwidth}
\centering\includegraphics[width=.66\textwidth]{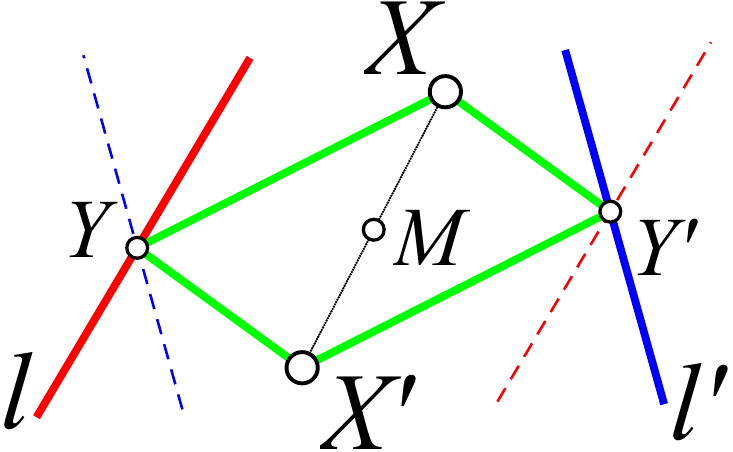}
\caption{Illustration of Claim~\ref{claim:XX'll'}}\label{fig:def-XXLL}
\end{minipage}
\begin{minipage}{0.73\textwidth}
\centering\includegraphics[width=.95\textwidth]{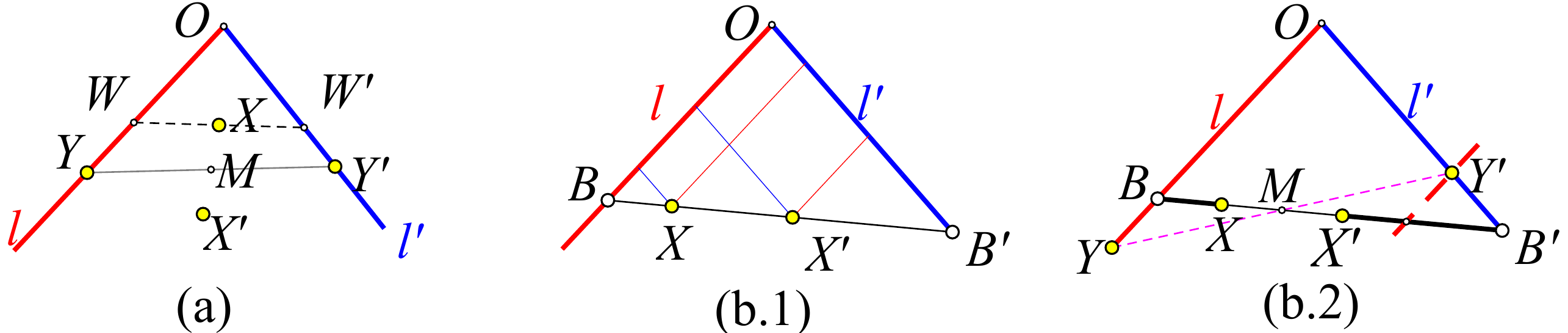}
\caption{Illustration of the proof of Lemma~\ref{lemma:area-pre}}\label{fig:area-pre}
\end{minipage}
\end{figure}

\begin{proof}
Let $M$ be the center of $\parallelogram(X,X',l,l')$.

(1) When $\dist_{l,l'}(X)<\dist_{l,l'}(X')$, there are three cases:
\begin{enumerate}
\item[Case~a:] \underline{$d_l(X)\leq d_l(X')$, $d_{l'}(X)\leq d_{l'}(X')$, and at least one inequality is strict.}
    See Figure~\ref{fig:area-pre}~(a).
    Let $(W,W')$ be the unique pair of points in $l,l'$ so that $\M(W,W')=X$.
    $
    \left\{\begin{array}{ccccccccc}
    d_{l'}(W) &=&2d_{l'}(X) &\leq & d_{l'}(X)+d_{l'}(X') &= &2d_{l'}(M)&=&d_{l'}(Y),\\
    d_l(W')   &=&2d_l(X)    &\leq & d_l(X)+d_l(X') &= &2d_l(M)&=&d_l(Y').
    \end{array}\right.
    $

    Therefore, $|OW|\leq |OY|$ and $|OW'|\leq |OY'|$, and at least one inequality is strict.
    Therefore, $X$ lies in $\triangle OYY'$ and $X\notin \overline{YY'}$, and so $X'\notin \triangle OYY'$. \smallskip

\item[Case~b:] \underline{$d_l(X)< d_l(X')$, $d_{l'}(X)>d_{l'}(X')$}. See Figure~\ref{fig:area-pre}~(b.1).
    Assume the extended line of $\overline{XX'}$ intersects $l,l'$ at $B,B'$ respectively.
    Using the condition of this case, $B,X,X',B'$ lie in this order.
    Applying $\dist_{l,l'}(X)<\dist_{l,l'}(X')$ and part (2b) of Lemma~\ref{lemma:dist_concave}, $|BX|<|X'B'|$.
    Now, see Figure~\ref{fig:area-pre}~(b.2). Since $|BX|<|X'B'|$ and $|MX|=|MX'|$, we get $|BM|<|B'M|$.
    So, $d_{l'}(B)<2d_{l'}(M)$, namely, $d_{l'}(B)<d_{l'}(Y)$, i.e.\ $|OB|<|OY|$.
    Further since $X$ lies in $\overline{BM}$, point $X$ lies in $\triangle OYY'$.
    Moreover, $X\neq M$ since $X\neq X'$. This further implies $X'\notin \triangle OYY'$. \smallskip

\item[Case~c:] \underline{$d_l(X)>d_{l'}(X')$ and $d_{l'}(X)< d_{l'}(X')$}. This case is symmetric to Case~b.
\end{enumerate}

In each case, we get $X\in \triangle OYY'$ and $X'\notin \triangle OYY'$.

\smallskip\noindent (2) This one is symmetric to (1); proof omitted.

\smallskip \noindent (3) It is similar to (1). We only consider Case~b here.
    In this case, we can get $|BX|=|B'X'|$. So, $|BM|=|MB'|$.
    By Claim~\ref{claim:ll'}, this means $B=Y$ and $B'=Y'$. Thus $X,X',Y,Y'$ lie in the same line, and so $\parallelogram(X,X',l,l')$ is degenerate.
 \end{proof}

\begin{lemma}\label{lemma:area}
Assume $l,l',X,X'$ are given as in Claim~\ref{claim:XX'll'}.
\footnote{For the equation in Lemma~\ref{lemma:area} to hold, the constraint ``$X,X'$ lie in the same quadrant'' is actually redundant.
  However, if we remove it, we should define the distance from a point to a line in a signed version,
    so that the points on different sides of a line have different signs.}
Let $\theta$ denote the angle of the quadrant divided by $l,l'$ and containing $X,X'$, then
    \begin{equation}\label{eqn:area}
    Area\left(\parallelogram(X,X',l,l')\right)=\left| \dist_{l,l'}(X)-\dist_{l,l'}(X')\right| / \sin \theta.
    \end{equation}
\end{lemma}

Below we give a totally geometric proof for this identity.

\begin{figure}[h]
  \centering\includegraphics[width=.88\textwidth]{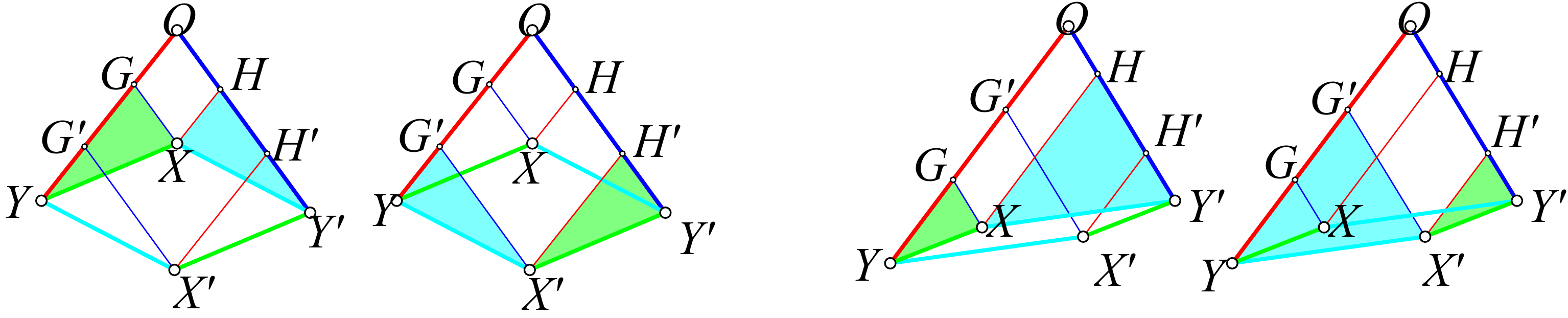}
  \caption{The geometric proof of Identity~(\ref{eqn:area}).}\label{fig:area_geometric}
\end{figure}

\begin{proof}
See Figure~\ref{fig:area_geometric}.
Let $G,H$ be the two points on $l$ and $l'$ such that $OGXH$ is a parallelogram,
and $G',H'$ the two points on $l$ and $l'$ such that $OG'X'H'$ is a parallelogram.
When $\dist_{l,l'}(X)=\dist_{l,l'}(X')$, by Lemma~\ref{lemma:area-pre}, $\parallelogram(X,X',l,l')$ is degenerate, so (\ref{eqn:area}) holds.
Next, assume $\dist_{l,l'}(X)<\dist_{l,l'}(X')$ (the other case where $\dist_{l,l'}(X)>\dist_{l,l'}(X')$ is symmetric).
We state two facts. (Fact (i) follows from Lemma~\ref{lemma:area-pre}, whereas (ii) is because $XYX'Y'$ is a parallelogram.)
\begin{enumerate}
\item[(i)] Point $X$ lies in the quadrilateral $OYX'Y'$.
\item[(ii)] $\triangle~GXY$ is congruent to $\triangle~H'Y'X'$ while $\triangle~HXY'$ is congruent to $\triangle~G'YX'$.
\end{enumerate}
Together, $Area(XYX'Y')\\
    =Area(OYX'Y')-Area(GXY)-Area(HXY')-Area(OGXH)\\
    =Area(OYX'Y')-Area(H'Y'X')-Area(G'YX')-Area(OGXH)\\
    =Area(OG'X'H')-Area(OGXH)$.

Notice that $d_l(X)=|XG|\cdot \sin \theta$ and $d_{l'}(X)=|XH|\cdot \sin \theta$. Therefore,
$$Area(OGXH)=|XG|\cdot |XH|\cdot \sin \theta=d_l(X)d_{l'}(X)/\sin \theta=\dist_{l,l'}(X)/\sin \theta.$$
Similarly, $Area(OG'X'H')=\dist_{l,l'}(X')/\sin \theta.$

Substituting the last two equations into the previous one, we obtain (\ref{eqn:area}).
 \end{proof}

\section{The constraints of the LMAPs}\label{sect:clamping}

This section presents two types of constraints of the LMAPs.
First, each LMAP has a corner lying on a vertex of $P$.
Second, for each corner of the LMAP,
  we distinguish $\Theta(n^2)$ situations and prove that under each situation
    this corner lies in a corresponding boundary-portion.
The situation depends on the locations of its two neighboring corners.
(The second type was previewed in subsection~\ref{subsect:techover}.)
All these constraints are deduced from the local maximality of the LMAPs (Definition~\ref{def:locallymaximal}).

\subsection{Each LMAP has a corner lying on a vertex of $P$}

The following lemma implies that each LMAP has a corner lying on a vertex of $P$.
\begin{lemma}\label{lemma:narrow_anchored}
Assume that $Q=A_0A_1A_2A_3$ is an LMAP, where $A_0,A_1,A_2,A_3$ lie in clockwise order.
Further assume that $A_3,A_1$ lie on $e_i,e_j$, respectively, such that $e_i\prec e_j$.
We claim that corner $A_0$ lies on a vertex of $P$.
\end{lemma}

\begin{proof}
For a contradiction, suppose $A_0$ does not lie on a vertex of $P$ but lies in edge $e_k$. See Figure~\ref{fig:PC0}.
Denote by $B$ the one among $\I_{i,k},\I_{j,k}$ which is closer to $A_0$; let $B$ be either of them for a tie.
Denote $Q_X=\parallelogram(X,A_2,\el_i,\el_j)$ and $d()=\dist_{\el_{i},\el_{j}}()$ for short.
Suppose $X$ is any point on $\overline{A_0B}$ but $A_0$. Then,
\[
\begin{array}{cl}
  d(A_2)>d(A_0), & \hbox{according to Lemma~\ref{lemma:area-pre};} \\
  d(A_0)>d(X), & \hbox{according to the concavity of $d()$ on $\overline{\I_{i,k}\I_{j,k}}$;}\\
  \begin{array}{c}
     Area(Q_X)=c\cdot |d(A_2)-d(X)|,\\
     Area(Q_{A_0})=c\cdot |d(A_2)-d(A_0)|,
   \end{array} & \hbox{due to (\ref{eqn:area}). Here, $c$ is a positive constant.}
\end{array}\]

Altogether, $Area(Q_X)>Area(Q_{A_0})$.
Moreover, since neither $A_1$ nor $A_3$ lies on a vertex of $P$, parallelogram $Q_X$ will be inscribed in $P$ when $X$ is sufficiently close to $A_0$.
Thus, there is an inscribed parallelogram $Q_X$ sufficiently close to $Q_{A_0}$ and is larger than $Q_{A_0}$.
Hence $Q=Q_{A_0}$ is not locally maximal and is not an LMAP.
 \end{proof}

\begin{figure}[h]
\begin{minipage}[b]{.48\textwidth}
  \centering\includegraphics[width=.5\textwidth]{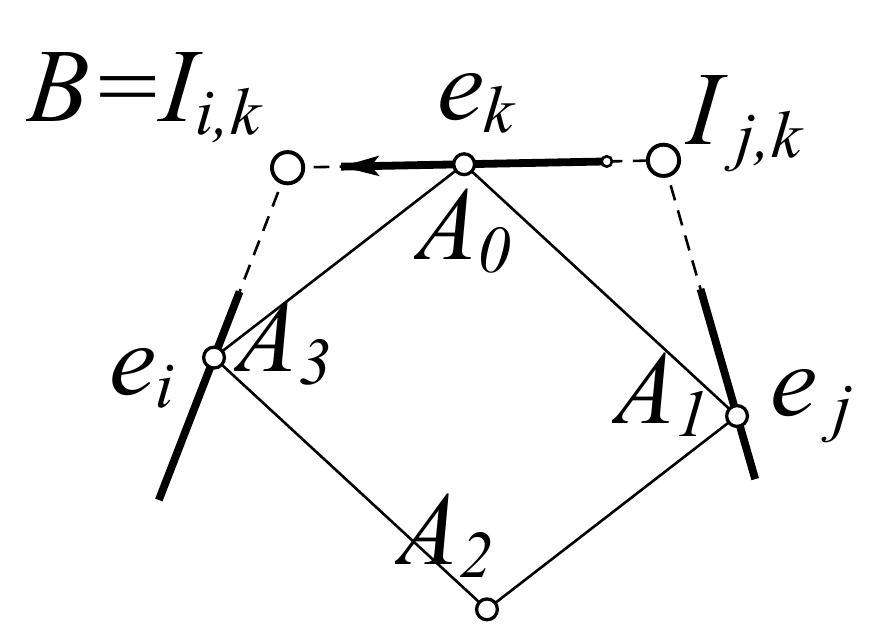}
  \caption{Illustration of Lemma~\ref{lemma:narrow_anchored}.}\label{fig:PC0}
\end{minipage}
\begin{minipage}[b]{.48\textwidth}
  \centering\includegraphics[width=.35\textwidth]{definition-chasing.pdf}
  \caption{Chasing relation between units.}\label{fig:chasing-u}
\end{minipage}
\end{figure}

\subsection{Description of clamping bounds of LMAPs}\label{subsect:clamping-description}

\newcommand{\unit}{\mathbf{u}}
To describe the the second type of constraints, we first introduce some terms and notations.
Recall that we call each vertex and each edge (of $P$) a \emph{\textbf{unit}} of $P$.

For any point $X$ on $\partial P$, there is a unique unit containing $X$, denoted by $\unit(X)$.

\begin{definition}[\textbf{Backward and forward edge of units and points}]\label{def:units}
Assume $u$ is a unit of $P$.
If $u$ is vertex $v_i$, its \emph{backward edge} and \emph{forward edge} is defined to be $e_{i-1}$ and $e_i$, respectively.
Otherwise, the \emph{backward edge} and \emph{forward edge} of $u$ is defined to be the edge $u$ itself.
Intuitively, when you start at any point in $u$ and move backward (forward) in clockwise along $\partial P$ by an infinitely small step, you will be located at the backward (forward) edge of $u$.
Denote the backward and forward edge of $u$ by $back(u)$ and $forw(u)$ respectively.
For convenience, we also define the backward and forward edge of point $X$. Specifically,
$$back(X):=back(\unit(X)); \quad forw(X):=forw(\unit(X)).$$
\end{definition}

\begin{definition}[\textbf{``Chasing'' relation between units}]
 We now extend the chasing relation $\prec$ given in section~\ref{sect:preliminiary} so that it is defined among units.
For two units $u,u'$, we say that $u$ \emph{is chasing} $u'$ if $back(u)\prec back(u')\text{ and }forw(u)\prec forw(u')$.
\end{definition}

The relation chasing between units is a compatible extension of chasing between edges.
Note that it is possible that neither of them is chasing the other, so there are three relations between two units $u,u'$:
1. $u$ is chasing $u'$ while $u'$ is not chasing $u$.
2. $u'$ is chasing $u$ while $u$ is not chasing $u'$. 3. Neither of them is chasing the other.

In Figure~\ref{fig:chasing-u}, $v_1$ is chasing $v_2,e_2,v_3,e_3$ whereas $e_4,v_5,\ldots,e_6,v_7$ are chasing $v_1$. For other units $e_1,e_7,v_4$, neither they are chasing $v_1$, nor $v_1$ is chasing them.

\begin{definition}[$\zeta$]\label{def:zeta}
Recall the $Z$-points lying on $\partial P$ introduced in section~\ref{sect:Z}.
For any unit pair $(u,u')$ such that \textbf{\emph{$u$ is chasing $u'$}}, we define a boundary-portion
    \begin{equation}\label{eqn:zeta_chasing}
        \zeta(u,u')= [Z_{back(u)}^{back(u')} \circlearrowright Z_{forw(u)}^{forw(u')}].
    \end{equation}
\end{definition}

We describe the second type of constraints of the LMAPs in two lemmas below.
We prove these two lemmas together in the next subsection.
In these lemmas, assume that $A_0,\ldots,A_3$ lie in clockwise order and subscripts of $A$ are taken modulo 4.

\begin{lemma}\label{lemma:clamping-broad}
Consider any corner $A_i$ of an LMAP $A_0A_1A_2A_3$.
Assume that $A_{i+1},A_{i-1}$ lie on units $u,u'$ respectively.
Then, $A_i$ lies in $\zeta(u,u')$ if $u$ is chasing $u'$.
\end{lemma}

\begin{figure}[h]
  \centering \includegraphics[width=.95\textwidth]{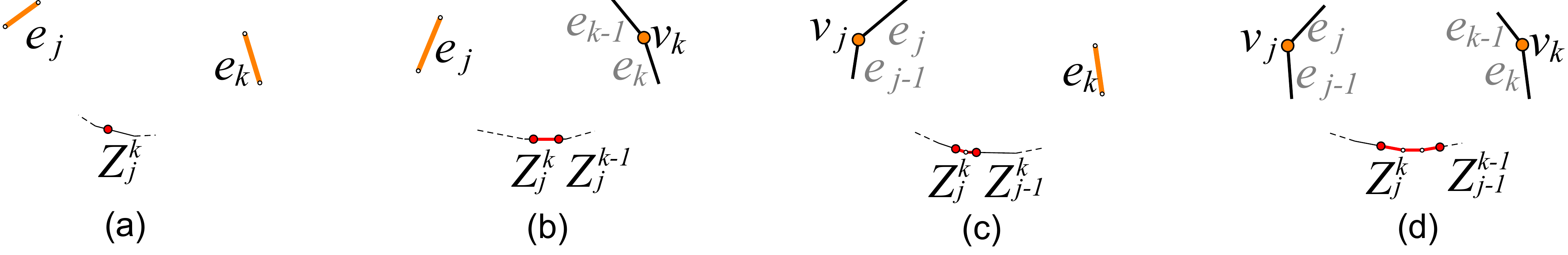}
  \caption{Illustration of $\zeta(u,u')$ when $u$ is chasing $u'$}\label{fig:BP}
\end{figure}

See Figure~\ref{fig:BP}. If $A_{i+1}\in u,A_{i-1}\in u'$, and $u$ is chasing $u'$, there are four cases:

1. $(u,u')=(e_j,e_k)$ and $e_j$ is chasing $e_k$. Then, $A_i=Z_j^k$.

2. $(u,u')=(e_j,v_k)$ and $e_j$ is chasing $v_k$. Then, $A_i\in [Z_j^{k-1} \circlearrowright Z_j^k]$.

3. $(u,u')=(v_j,e_k)$ and $v_j$ is chasing $e_k$. Then, $A_i\in [Z_{j-1}^{k}\circlearrowright Z_j^k]$.

4. $(u,u')=(v_j,v_k)$ and $v_j$ is chasing $v_k$. Then, $A_i\in [Z_{j-1}^{k-1}\circlearrowright Z_j^k]$.

\bigskip In the above lemma, a bound is given for $A_i$ when $u=\unit(A_{i+1})$ is chasing $u'=\unit(A_{i-1})$.
 In the next lemma, a bound will be given for $A_i$ under the case neither of $u,u'$ is chasing the other.
 To state the next lemma, we extend the scope of definition of $\zeta(u,u')$ to every pair of distinct units $u,u'$.

\begin{remark}
As introduced in subsection~\ref{subsect:techover}, defining $\zeta(u,u')$ for bounding $A_i$ is not easy.
  This is particularly challenging when $u=\unit(A_{i+1})$ is not chasing $u'=\unit(A_{i-1})$.
The reader may first have a try before looking at the following definition.
\end{remark}

\begin{definition}[Generalized $\zeta$]\label{def:zeta-extended}
Recall $\D_i$: the unique vertex of $P$ with the largest distance to $\el_i$.
For each pair of units $u,u'$ that are distinct, we define
\begin{equation}\label{eqn:zeta_generalized}
\zeta(u,u') = [Z_{a}^{a'}\circlearrowright Z_{b}^{b'}],
~where
\left\{\begin{array}{lcl}
e_a & = & back(u) \\
e_{a'} & = & \begin{cases}
            back(u'), & \hbox{if } back(u) \prec back(u'); \\
            back(\D_a), & \hbox{otherwise}.
          \end{cases}\\
e_{b'} & = & forw(u')\\
e_b & = & \begin{cases}
            forw(u), & \hbox{if } forw(u) \prec forw(u'); \\
            forw(\D_{b'}), & \hbox{otherwise}.
        \end{cases}
\end{array}\right.
\end{equation}
\end{definition}

Be aware that it is always true that $e_a\prec e_{a'}$ and $e_b\prec e_{b'}$, so $\zeta(u,u')$ is well-defined.
Also be aware that (\ref{eqn:zeta_generalized}) simplifies to (\ref{eqn:zeta_chasing}) when $u$ is chasing $u'$.

\begin{figure}[h]
  \centering\includegraphics[width=.6\textwidth]{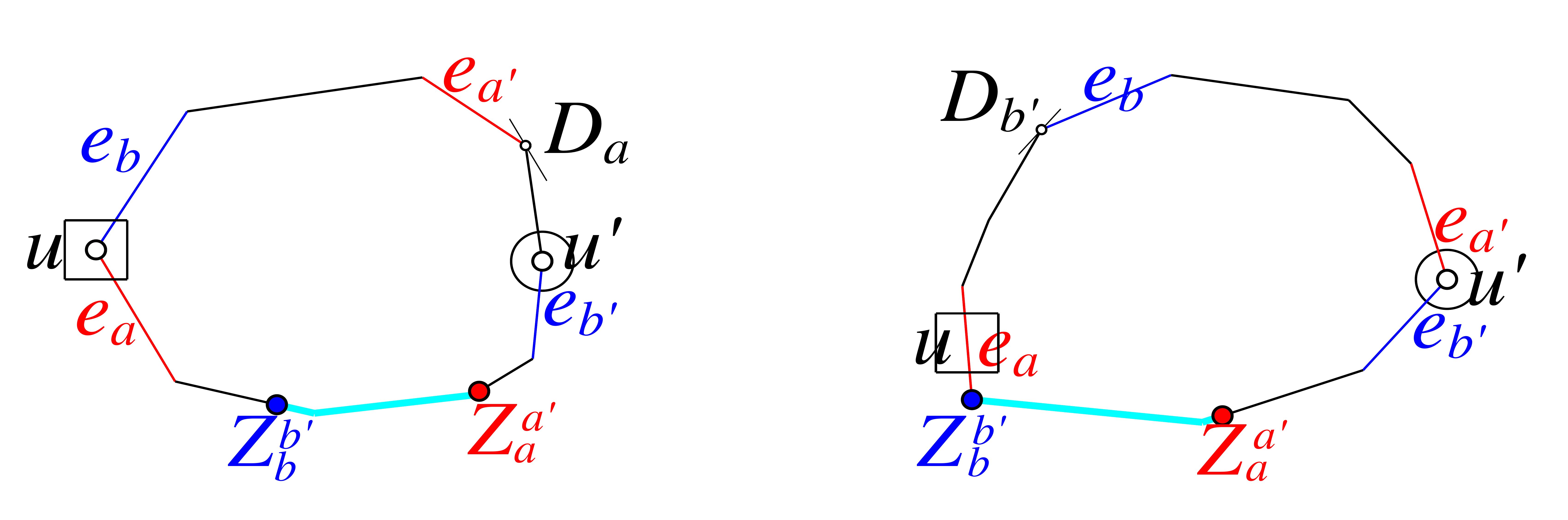}
  \caption{Illustration of $\zeta(u,u')$ when neither of $u,u'$ is chasing the other.}\label{fig:zeta2}
\end{figure}

\begin{lemma}\label{lemma:clamping-even}
Under the same assumption as Lemma~\ref{lemma:clamping-broad},
$A_i$ lies in $\zeta(u,u')$ if neither of $u,u'$ is chasing the other.
(This lemma applies he boundary-portions in $\{\zeta(u,u')\mid \text{neither of }u,u' \text{ is chasing the other}\}$ which are illustrated in Figure~\ref{fig:zeta2}.)
\end{lemma}

\begin{remark}
1. Even in the last case where $u'$ is chasing $u$, we can prove $A_i\in \zeta(u,u')$;
therefore, no matter what the chasing relation between $u$ and $u'$, we call $\zeta(u,u')$ the \textbf{\emph{clamping bound}} of $A_i$.
However, in the last case, the bound $\zeta(u,u')$ is too wide and useless; so we do not state it formally in another lemma.

2. To design our quadratic time algorithm for computing the LMAPs, we only apply Lemma~\ref{lemma:clamping-broad} but \textbf{not} Lemma~\ref{lemma:clamping-even}.
However, to design a better (subquadratic time) algorithm (in our follow-up paper), we must also apply Lemma~\ref{lemma:clamping-even}.
We find it is more convenient to prove Lemma~\ref{lemma:clamping-even} together with Lemma~\ref{lemma:clamping-broad} in this paper,
so that the follow-up paper does not need to go through this proof technique again.
Moreover, it is easier to compare these two related lemmas in the same paper.
\end{remark}

\subsection{Proofs of the clamping bounds}

We first sketch the proof of Lemma~\ref{lemma:clamping-broad} so that the reader can grasp our key ideas.
\begin{figure}[h]
  \centering \includegraphics[width=.64\textwidth]{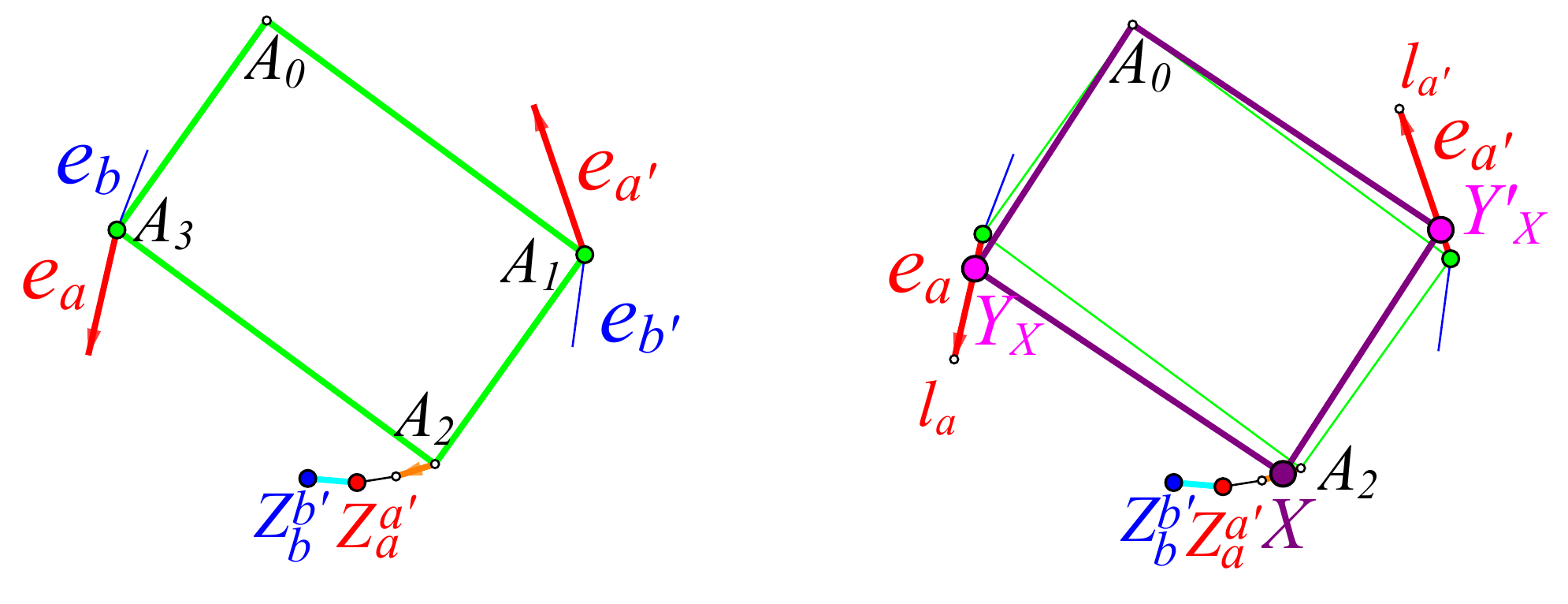}
  \caption{Illustration of the proof of the clamping bounds.}\label{fig:PC_sketch}
\end{figure}

\begin{proof}[Sketch]
Assume that $Q=A_0A_1A_2A_3$ is an LMAP and that $A_3\in u,A_1\in u'$, where $u$ is chasing $u'$. See Figure~\ref{fig:PC_sketch}.
Denote $back(u),back(u'),forw(u),forw(u')$ by $e_a,e_{a'},e_b,e_{b'}$ respectively.
Notice that $e_{a}\prec e_{a'}$ and $e_b\prec e_{b'}$, since $u$ is chasing $u'$.
We shall prove that corner $A_2$ lies in $[Z_a^{a'}\circlearrowright Z_b^{b'}]$.

For a contradiction, suppose that $A_2\notin [Z_a^{a'}\circlearrowright Z_b^{b'}]$.
Then, it must lie in $(A_1\circlearrowright Z_a^{a'})$ or $(Z_b^{b'}\circlearrowright A_3)$.
Assume it lies in $(A_1\circlearrowright Z_a^{a'})$; otherwise it is symmetric.

For any point $X\in[A_2\circlearrowright Z_a^{a'}]$, denote $Q_X=\parallelogram(A_0,X,\el_a,\el_{a'})$ and
  let $(Y_X,Y'_X)$ denote the opposite pair of corners of $Q_X$ that are restricted to $\el_a,\el_{a'}$.

We state three observations.
\begin{enumerate}
\item[(i)] \emph{$Area(Q_X)$ is proportional to $\dist_{\el_a,\el_{a'}}(X)-\dist_{\el_a,\el_{a'}}(A_0)$.}
\item[(ii)] \emph{$\dist_{\el_a,\el_{a'}}(X)$ strictly increases when $X$ moves along $[A_2\circlearrowright Z_a^{a'}]$.} Note:
``move along $[A_2\circlearrowright Z_a^{a'}]$'' is short for ``move in clockwise along $\partial P$ from $A_2$ to $Z_a^{a'}$''.
\item[(iii)] \emph{We have ($Y_X\in e_a$ and $Y'_X\in e_{a'}$) when $X$ is sufficiently close to $A_2$.}
\end{enumerate}

(i) is due to Lemma~\ref{lemma:area}; (ii) is an application of the unimodality of the product-distance function (Lemma~\ref{lemma:dist_unimodal});
  and (iii) is trivial. We omit their proofs in this sketch.

\smallskip Combining (i) and (ii), $Area(Q_X)$ strictly increases when $X$ moves along $[A_2\circlearrowright Z_a^{a'}]$ (starting from $A_2$).
By (iii), $Q_X$ is inscribed in $P$ when $X$ is sufficiently close to $A_2$. Together, $Q_{A_2}$ (i.e.\ $Q$) is not locally maximal.
So $A_2\in [Z_a^{a'}\circlearrowright Z_b^{b'}]$.
\end{proof}

\begin{proof}[Proof of Lemmas~\ref{lemma:clamping-broad}, \ref{lemma:clamping-even}]
Let $Q=A_0A_1A_2A_3$ be an LMAP as specified by Lemma~\ref{lemma:clamping-broad} or Lemma~\ref{lemma:clamping-even}.
Consider corner $A_2$.
Let $u=\unit(A_3),u'=\unit(A_1)$ and let $a,a',b,b'$ be defined according to (\ref{eqn:zeta_generalized}).
We shall prove that $A_2\in [Z_a^{a'}\circlearrowright Z_b^{b'}]$ when
\begin{equation}\label{eqn:broad-or-even}
\text{$u$ is chasing $u'$ or neither of $u,u'$ is chasing the other.}
\end{equation}

Generally, we proceed by contradiction. If $A_2$ does not lie in $\zeta(u,u')$, we want to construct a strictly larger parallelogram $Q'$ arbitrarily close to $Q$.
As mentioned in subsection~\ref{subsect:techover}, we construct $Q'$ by \emph{(slightly) changing the position of $A_2$ while fixing its opposite corner and adjusting the other corners accordingly within $\partial P$}.\smallskip

First, we state three key arguments and deduce the clamping bounds from them.
\begin{enumerate}
\item[(i)] At least one of the points $Z_a^{a'},Z_b^{b'}$ lies in $(A_1\circlearrowright A_3)$.
\item[(ii)] When point $Z_a^{a'}$ lies in $(A_1 \circlearrowright A_3)$, corner $A_2\notin (A_1 \circlearrowright Z_a^{a'})$.
\item[(iii)] When point $Z_b^{b'}$ lies in $(A_1 \circlearrowright A_3)$, corner $A_2\notin (Z_b^{b'} \circlearrowright A_3)$.
\end{enumerate}

Since $A_2$ always lies in $(A_1\circlearrowright A_3)$, using (i), (ii) and (iii) we can obtain $A_2\in [Z_a^{a'}\circlearrowright Z_b^{b'}]$.
To see this more clearly, consider whether both $Z_a^{a'},Z_b^{b'}$ or only one of them lies in $\rho=(A_1\circlearrowright A_3)$.
Let us assume that both $Z$-points lie in $\rho$; the other case is easier and similar.
Now, there are two subcases: $Z_a^{a'}\leq_\rho Z_b^{b'}$ or $Z_b^{b'}<_\rho Z_a^{a'}$,
  as shown in Figure~\ref{fig:PC1}~(a),(b).
In the former subcase, by (ii) and (iii), $A_2$ can only lie in $[Z_a^{a'}\circlearrowright Z_b^{b'}]$.
In the latter subcase, by (ii) and (iii), $A_2\notin (A_1\circlearrowright A_3)$; so actually this subcase would not happen.
(It cannot happen indeed due to the bi-monotonicity of the $Z$-points; see Lemma~\ref{lemma:Z_bi-monotonicity} and see also Figure~\ref{fig:zeta2}.)

\begin{figure}[h]
\centering\includegraphics[width=.75\textwidth]{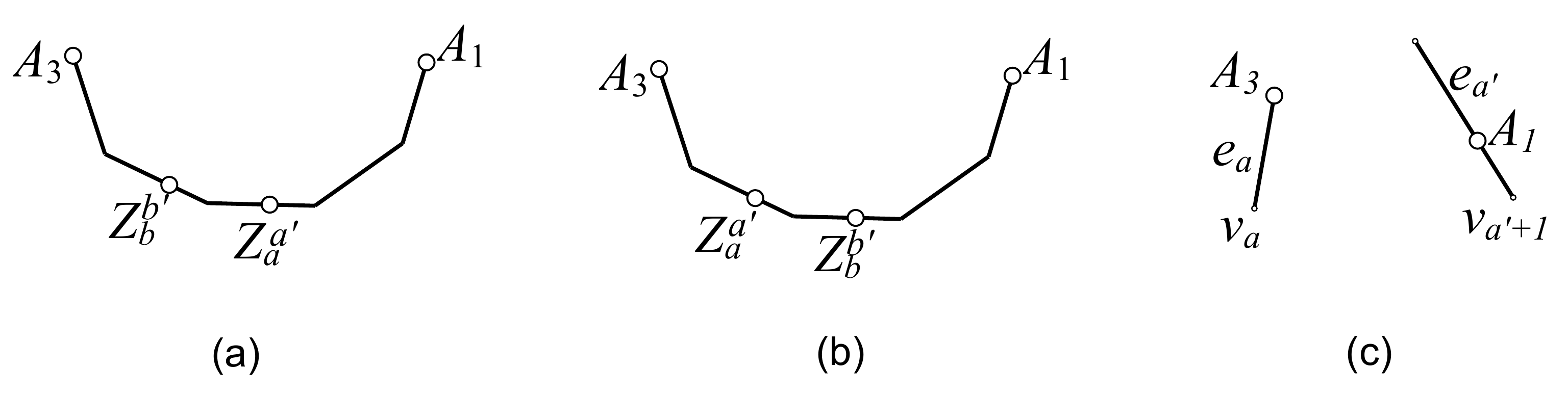}
\caption{Proofs of the clamping bounds - preliminary part}\label{fig:PC1}
\end{figure}

\smallskip We state one more argument before we prove (i), (ii) and (iii).
\begin{enumerate}
\item[(*)] $back(u')\neq back(u)$ and $forw(u')\neq forw(u)$.
\end{enumerate}

\noindent
Proof of (*):  Since $A_1,A_3$ are opposite corners of an inscribed parallelogram,
$back(A_1)\neq back(A_3)$ and $forw(A_1)\neq forw(A_3)$; this implies (*).

\medskip \noindent
Proof of (i). By (\ref{eqn:broad-or-even}), $u'$ is not chasing $u$.
This means $back(u')\nprec back(u)$ or $forw(u')\nprec forw(u)$.
Further since (*), we get $back(u)\prec back(u')$ or $forw(u)\prec forw(u')$.
Therefore, the following observations easily imply that at least one point in $\{Z_a^{a'},Z_b^{b'}\}$ lies in $(A_1 \circlearrowright A_3)$.
\begin{enumerate}
\item[(i.1)] If $back(u)\prec back(u')$, point $Z_a^{a'}$ lies in $(A_1 \circlearrowright A_3)$.
\item[(i.2)] If $forw(u)\prec forw(u')$, point $Z_b^{b'}$ lies in $(A_1 \circlearrowright A_3)$.
\end{enumerate}

We only prove (i.1); (i.2) is symmetric.

Assume that $back(u)\prec back(u')$. Then, due to (\ref{eqn:zeta_generalized}),
$e_{a'}=back(u')=back(A_1)$ and $e_a=back(u)=back(A_3)$. So $(v_{a'+1}\circlearrowright v_a)\subseteq (A_1 \circlearrowright A_3)$, as illustrated in Figure~\ref{fig:PC1}~(c).
However, $Z_a^{a'}\in (v_{a'+1}\circlearrowright v_a)$ by Lemma~\ref{lemma:dist-unique-location}.
So, $Z_a^{a'}\in (A_1 \circlearrowright A_3)$.

\medskip \noindent We prove (ii) in the following. The proof of (iii) is symmetric and omitted.

For a contradiction, suppose that $Z_{a}^{a'}\in (A_1 \circlearrowright A_3)$ and $A_2\in(A_1 \circlearrowright Z_{a}^{a'})$.
We shall show that $Q$ is not locally maximal.

We have to discuss two different cases.
Notice that $back(u)\neq back(u')$ by (*), so there are two cases:
$back(u)\prec back(u')$, or $back(u')\prec back(u)$.

\smallskip \noindent \textbf{Case~1:} $back(u)\prec back(u')$. See Figure~\ref{fig:PC2}~(a).

In this case, $e_a=back(u)=back(A_3)$ and $e_{a'}=back(u')=back(A_1)$ by (\ref{eqn:zeta_generalized}).
Take a point $B$ from the intersection of $forw(A_2)$ and $(A_2\circlearrowright Z_a^{a'})$.
Let point $X$ be restricted to segment $\overline{A_2B}$ and distinct from $A_2$.
Denote $Q_X=\parallelogram(X,A_0,\el_a,\el_{a'})$ and $d()=\dist_{\el_{a},\el_{a'}}()$.
The following imply that $Q=Q_{A_2}$ is not locally maximal.

(I) $Area(Q_X)>Area(Q_{A_2})$.

(II) $Q_X$ is inscribed in $P$ when $X$ is sufficiently close to $A_2$.

\begin{figure}[h]
\centering\includegraphics[width=.9\textwidth]{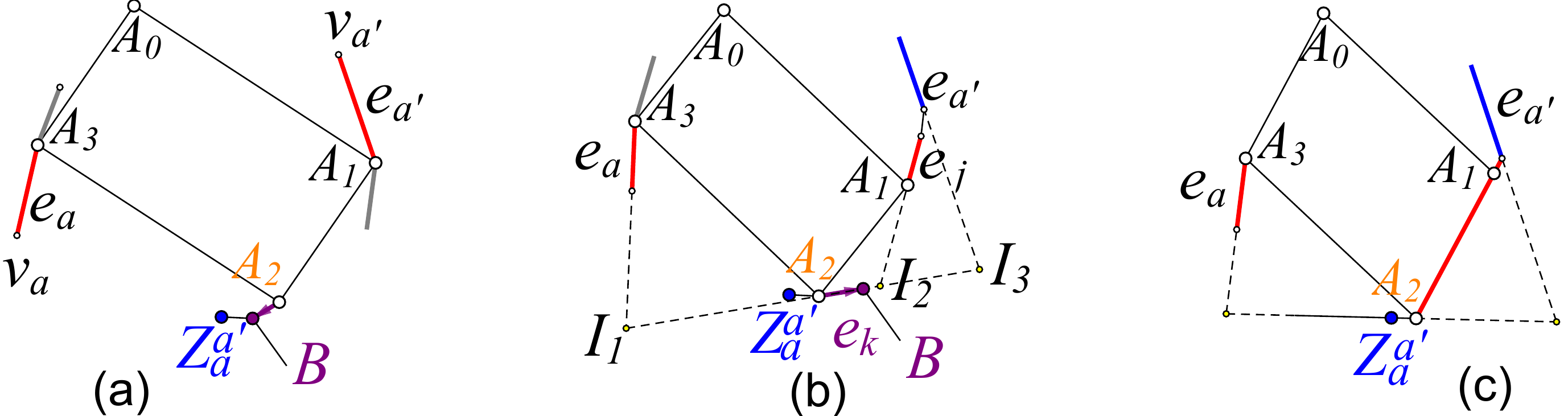}
\caption{Proofs of the clamping bounds - essential part.}\label{fig:PC2}
\end{figure}

\smallskip \noindent Proof of (I): It follows from the following facts.
$$
\begin{array}{cl}
  d(A_2)>d(A_0), & \hbox{since $e_a\prec e_{a'}$ and according to Lemma~\ref{lemma:area-pre};} \\
  d(X)>d(A_2), & \hbox{according to unimodality of $d()$ (Lemma~\ref{lemma:dist_unimodal});}\\
  \begin{array}{c}
     Area(Q_X)=c\cdot |d(X)-d(A_0)|,\\
     Area(Q_{A_2})=c\cdot |d(A_2)-d(A_0)|,
   \end{array} & \hbox{due to (\ref{eqn:area}). Here, $c$ is a positive constant.}
\end{array}$$

\noindent Proof of (II): Notice that $Q_X=\parallelogram(X,A_0,\el_a,\el_{a'})$ has a corner lying in $\el_a$ and a corner in $\el_{a'}$.
We shall prove that when $X$ moves straight from $A_2$ towards $B$,

(II.1) the corner of $Q_X$ restricted to $\el_a$ moves toward $v_a$; and

(II.2) the corner of $Q_X$ restricted to $\el_{a'}$ moves toward $v_{a'}$.

We prove (II.1) in the following; (II.2) is symmetric.

Because $A_1,A_2$ are neighboring corners of $Q$, we get $back(A_1)\prec forw(A_2)$, i.e.\ $e_{a'}\prec forw(A_2)$.
Therefore, $X$ gradually gets away from $\el_{a'}$ during its movement from $A_2$ to $B$.
So, the center of $Q_X$ gradually gets away from $\el_{a'}$, since it moves in the same direction as $X$.
So, the reflection of $\el_{a'}$ around the center of $Q_X$ gradually gets away from $\el_{a'}$;
i.e.\ the corner of $Q_X$ inscribed in $\el_a$ gradually gets away from $\el_{a'}$.
This implies (II.1) since $e_a\prec e_{a'}$.

\bigskip \noindent \textbf{Case~2:} $back(u')\prec back(u)$. See Figure~\ref{fig:PC2}~(b).

We first state that $back(A_2)\neq back(A_1)$; its proof is deferred for a moment.

Denote $e_j=back(A_1),e_k=back(A_2)$.
Let $B$ be any point in $e_k$ but not in $[A_2\circlearrowright Z_a^{a'}]$.
Let point $X$ be restricted to segment $\overline{A_2B}$ and distinct from $A_2$.
Denote $Q_X=\parallelogram(X,A_0,\el_a,\el_j)$.
Assume that $\el_k$ intersects $\el_a,\el_j,\el_{a'}$ at $\I_1,\I_2,\I_3$, respectively.\smallskip

Applying the unimodality of $\dist_{\el_a,\el_{a'}}$ (Lemma~\ref{lemma:dist_unimodal}), $\dist_{\el_a,\el_{a'}}(X)$ strictly decreases when $X$ moves straight from $A_2$ to $B$. This implies that $|A_2\I_3|\leq |A_2\I_1|$ according to Lemma~\ref{lemma:dist_concave}, which further implies that $|A_2\I_2|\leq |A_2\I_1|$. Apply the last inequality and Lemma~\ref{lemma:dist_concave} again, $\dist_{\el_a,\el_j}(X)$ decreases strictly when $X$ moves straight from $A_2$ to $B$. So, $\dist_{\el_a,\el_j}(X)<\dist_{\el_a,\el_j}(A_2)$.
(Be aware of the different subscripts of $\dist$.)

\smallskip Notice that $e_j\prec e_a$ since $back(u')\prec back(u)$. Therefore, by Lemma~\ref{lemma:area-pre} and Lemma~\ref{lemma:area}, $Area(Q_X)$ is proportional to $\dist_{\el_a,\el_j}(A_0)-\dist_{\el_a,\el_j}(X)$.

Combining the above two results, we obtain $Area(Q_X)>Area(Q_{A_2})$.

In addition, we claim that $Q_X$ is inscribed in $P$ when $X$ is sufficiently close to $A_2$. The proof is similar to that of Case~1 and hence omitted.

Together, $Q=Q_{A_2}$ is not locally maximal.

\medskip Finally, we verify $back(A_2)\neq back(A_1)$ stated above.
Suppose to the contrary that $back(A_2)=back(A_1)$, as shown in Figure~\ref{fig:PC2}~(c).
Since $back(A_2)=back(A_1)$, we get $forw(A_1)=back(A_2)$.
Since $A_2,A_3$ are neighboring corners of an inscribed parallelogram, we get $back(A_2)\prec forw(A_3)$.
Combining these two formulas, $forw(A_1)\prec forw(A_3)$, namely, $forw(u')\prec forw(u)$.
Further since $back(u')\prec back(u)$, unit $u'$ is chasing $u$, which contradicts assumption (\ref{eqn:broad-or-even}).
 \end{proof}

\section{Algorithm(s) for computing the LMAPs}\label{sect:algo}

This section demonstrates an $O(n^2)$ time algorithm for computing the LMAPs which utilizes the $\Theta(n^2)$ clamping bounds shown in the last section (Lemma~\ref{lemma:clamping-broad}).

\smallskip \noindent \textbf{A clarification.} Our algorithm will output many candidates of the LMAPs.
  Each LMAP is a candidate and thus all LMAPs will be outputted, yet not every candidate is an LMAP.
  For simplicity, we do not check whether a candidate is an LMAP or not,
    since it is unnecessary -- we can find the MAP anyway by choosing the largest candidate.
    (Nevertheless, checking the local maximality only takes $O(1)$ time since only a small amount of local information of $P$ matter in such a check.)

\begin{definition}[Classification of corners of inscribed parallelograms]\label{def:classification}
Assume a parallelogram $A_0A_1A_2A_3$ is inscribed in $\partial P$, and $A_0,A_1,A_2,A_3$ lie in clockwise order.
We classify every corner $A_i$ as narrow, broad or even:
\begin{enumerate}
\item[] \textbf{\emph{narrow}}: if $\unit(A_{i-1})$ is chasing $\unit(A_{i+1})$. (subscripts are taken modulo 4)
\item[]  \textbf{\emph{broad}}: if its opposite corner $A_{i+2}$ is narrow; i.e., if $\unit(A_{i+1})$ is chasing $\unit(A_{i-1})$.
\item[]  \textbf{\emph{even}}: if otherwise; i.e.\ if neither of $\unit(A_{i+1})$ and $\unit(A_{i-1})$ is chasing the other.
\end{enumerate}
\end{definition}

\begin{figure}[h]
\centering\includegraphics[width=.2\textwidth]{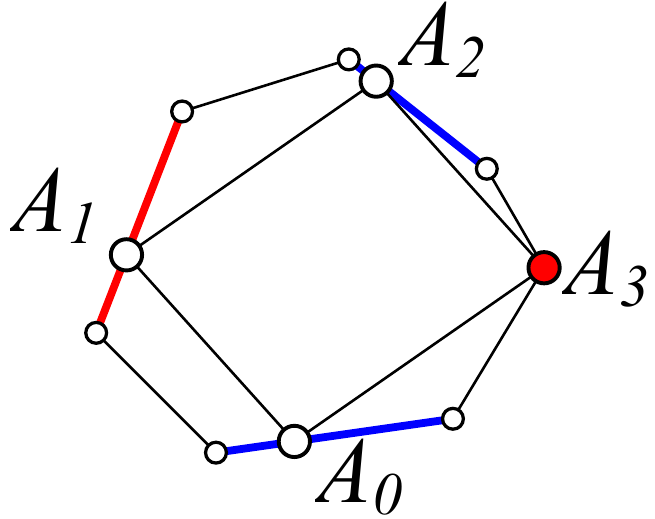}
\setcaptionwidth{.85\textwidth}
\caption{Illustration of broad, narrow and even corners.
In this example, unit $\unit(A_2)$ is chasing $\unit(A_0)$,
so $A_1$ is broad and $A_3$ is narrow.
Neither of $\unit(A_1),\unit(A_3)$ is chasing the other,
so $A_0$ and $A_2$ are both even.}\label{fig:threetypes}
\end{figure}

\begin{remark}
With the terms introduced in the above classification,
  we can observe that Lemma~\ref{lemma:clamping-broad} bounds broad corners, whereas Lemma~\ref{lemma:clamping-even} bounds even corners.
\end{remark}

A corner of an LMAP is \textbf{\emph{anchored}} if it coincides with a vertex of $P$.
Lemma~\ref{lemma:narrow_anchored} implies that an LMAP has an anchored corner and it in fact implies the following.

\begin{lemma}\label{lemma:3LMAPs}
For any LMAP $Q$, at least one of the following holds.
(a) $Q$ has an anchored narrow corner.
(b) $Q$ has an anchored broad corner that has at least one adjacent corner anchored.
(c) $Q$ has four even corners.
\end{lemma}

\begin{proof}
Assume $Q=A_0,A_1,A_2,A_3$, where $A_0,A_1,A_2,A_3$ lie in clockwise order.

First, suppose a pair of opposite corners of $Q$ are both unanchored.
Assume that $A_3,A_1$ are unanchored and edge $\unit(A_3)$ is chasing edge $\unit(A_1)$ (Figure~\ref{fig:PC0}).
Since $A_0$ is anchored by Lemma~\ref{lemma:narrow_anchored} and is narrow by definition~\ref{def:classification}, we get (a).
Next, assume that (X) among each pair of opposite corners, at least one corner is anchored.

If neither of $\unit(A_0),\unit(A_2)$ is chasing the other and so do the pair $\unit(A_1)$ and $\unit(A_3)$, all the corners are even and (c) holds.
Next, assume $\unit(A_2)$ is chasing $\unit(A_0)$.
This means $A_3$ is narrow and $A_1$ is broad by Definition~\ref{def:classification}.
Further, since at least one of $A_1,A_3$ is anchored by the assumption (X), $Q$ has an anchored narrow corner or an anchored broad corner (or both).
Further applying (X), (a) or (b) holds.
 \end{proof}

Our general algorithm consists of three algorithms.
One computes those LMAPs with \emph{one anchored narrow corner}.
One computes those with \emph{one anchored broad corner that has at least one adjacent corner anchored}.
Another computes those with \emph{four even corners}.
They are sufficient for computing all the LMAPs due to Lemma~\ref{lemma:3LMAPs}.
We will see each of them runs in $O(n^2)$ time, so the general algorithm runs in $O(n^2)$ time.
The first two algorithms use similar ideas -- they both apply the clamping bounds (Lemma~\ref{lemma:clamping-broad}).
The last algorithm is straightforward.

\subsection{Definition of blocks and some related bounds of the LMAPs}\label{subsect:blocks-transform}

We now deduce new bounds for corners of LMAPs from the clamping bounds
  by using the geometric fact that two diagonals of a parallelogram bisect each other.

\begin{definition}[Region $u \oplus u'$ for two units $u,u'$]
Recall that $\M(X,X')$ denotes the mid point of $X,X'$.
For distinct units $u,u'$, denote
\begin{equation}\label{def:midregion}
u\oplus u'=\{\M(X,X')\mid X\in u, X'\in u'\}.
\end{equation}
The shape of $u\oplus u'$ is a parallelogram, a segment, or a point.
More specifically, $e_i\oplus e_j$ is an open parallelogram, whose four corners are respectively $\M(v_i,v_j)$, $\M(v_i,v_{j+1})$, $\M(v_{i+1},v_j)$, and $\M(v_{i+1},v_{j+1})$;
region $e_i\oplus v_j$ is the open segment $\overline{\M(v_i,v_j)\M(v_{i+1},v_j)}$;
region  $v_i\oplus e_j$ is the open segment $\overline{\M(v_i,v_j)\M(v_i,v_{j+1})}$;
and $v_i\oplus v_j$ is the single point $\M(v_i,v_j)$. See Figure~\ref{fig:block_def}.
\end{definition}

\begin{definition}[Blocks]
Recall reflection and scaling in section~\ref{sect:preliminiary}.
Assume $u$ is chasing $u'$.
Recall $\zeta(u,u')$ in Definition~\ref{def:zeta}.
We consider four cases.
\begin{enumerate}
\item $(u,u')=(e_i,e_j)$. See Figure~\ref{fig:block_def}~(a).\\
    The $2$-scaling of $e_i\oplus e_j$ about point $Z_i^j$ is a parallelogram whose sides are congruent to either $e_i$ or $e_j$. We define $\block(e_i,e_j)$ to be this parallelogram.

\item $(u,u')=(v_i,v_j)$. See Figure~\ref{fig:block_def}~(d).\\
    The reflection of $\zeta(v_i,v_j)$ around $\M(v_i,v_j)$ is a polygonal curve, and we refer to it as $\block(v_i,v_j)$.

\item $(u,u')=(v_i,e_j)$. See Figure~\ref{fig:block_def}~(b).\\
    In this case, $\block(v_i,e_j)$ is the region bounded by four curves:\\
    the $2$-scaling of $v_i\oplus e_j$ about $Z_{i-1}^j$; the $2$-scaling of $v_i\oplus e_j$ about $Z_i^j$;
    the reflection of $\zeta(v_i,e_j)$ around $\M(v_i,v_j)$; the reflection of $\zeta(v_i,e_j)$ around $\M(v_i,v_{j+1})$.\\
    Observe that $\block(v_i,e_j)$ consists of parallelograms that are parallel to $e_j$.

\item $(u,u')=(e_i,v_j)$. See Figure~\ref{fig:block_def}~(c).\\
    We define $\block(e_i,v_j)$ symmetric to $\block(v_i,e_j)$.\\
    Observe that $\block(e_i,v_j)$ consists of parallelograms that are parallel to $e_i$.
\end{enumerate}
In this way, we define a set of $\Theta(n^2)$ regions $\{\block(u,u')\mid u \text{ is chasing }u'\}$ in the plane.
 For convenience, we call each such region a \emph{block}.
\end{definition}

\begin{figure}[t]
\centering\includegraphics[width=\textwidth]{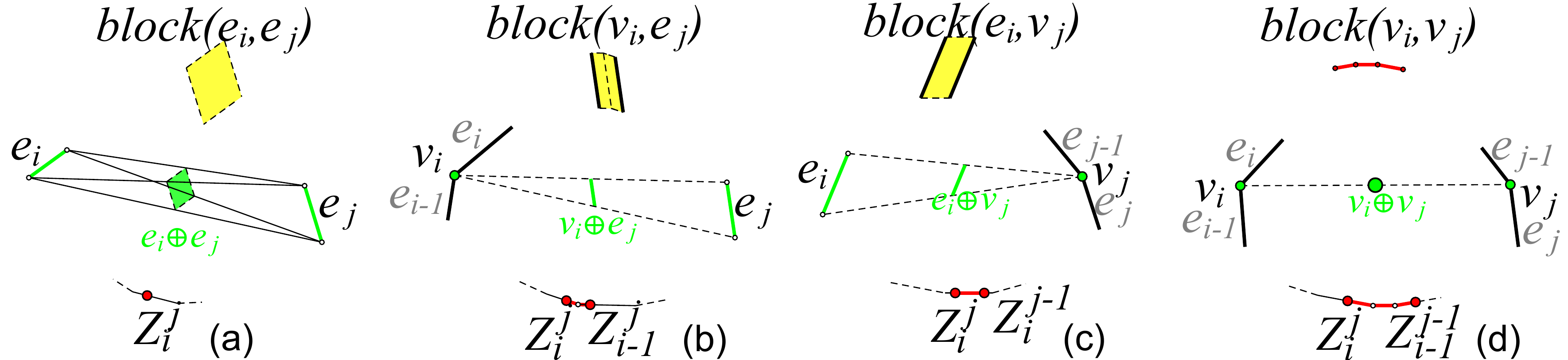}
\caption{Illustration of the geometric definition of the blocks.}\label{fig:block_def}
\end{figure}

\begin{lemma}\label{lemma:transformed}
Assume $A_0A_1A_2A_3$ is an LMAP, where $A_0,A_1,A_2,A_3$ lie in clockwise order, and $\unit(A_{i-1})$ is chasing $\unit(A_{i+1})$,
which means $A_i$ is narrow by Definition~\ref{def:classification}.
Then, $A_i$ lies in $\block(u,u')$, where $u=\unit(A_{i-1})$ and $u'=\unit(A_{i+1})$.
\end{lemma}
\begin{proof}
We only discuss the case where $u,u'$ are edges, e.g.\ $u=e_i,u'=e_j$. The other cases are similar.
See Figure~\ref{fig:block_def}~(a).
Since $A_{i-1}$ and $A_{i+1}$ lie in $u,u'$ respectively, $\M(A_{i-1},A_{i+1})$ lies in $u\oplus u'= e_i\oplus e_j$.
Since two diagonals of $Q$ bisect each other, $\M(A_{i-1},A_{i+1})=\M(A_i,A_{i+2})$.
Together, the mid point of $A_i,A_{i+2}$ lies in $(e_i\oplus e_j)$.
Applying Lemma~\ref{lemma:clamping-broad}, $A_{i+2}$ lies in $\zeta(u,u')$.
In other words, $A_{i+2}=Z_i^j$.
Together, $A_i$ lies in the 2-scaling of region $(e_i\oplus e_j)$ about $Z_i^j$,
i.e., $A_i\in \block(u,u')$.
 \end{proof}

\subsection{Compute the LMAPs with an anchored narrow corner}

\textbf{Note:} In the rest of this section, we always assume that $Q=A_0A_1A_2A_3$ is an LMAP (to be determined)
 and $A_0,A_1,A_2,A_3$ lie in clockwise order.

\begin{claim}\label{claim:N}
Suppose $A_3\in u$ and $A_1\in u'$, where $(u,u')$ is a given pair of units such that $u$ is chasing $u'$.
Suppose the position (which means the coordinates) of $A_0$ is given.
Then the positions of four corners of $Q$ are all determined.
\end{claim}

\begin{proof}
First, $A_2$ can be determined as follows.
If $(u,u')=(e_i,e_j)$, point $A_2$ lies at $Z_i^j$.
If $(u,u')=(v_i,v_j)$, point $A_2$ lies at the reflection of $A_0$ around $\M(v_i,v_j)$.
If $(u,u')=(v_i,e_j)$, point $A_2$ lies at the intersection of $s$ and $\zeta(u,u')$, where $s$ is the extended line of the 2-scaling of $(v_i\oplus e_j)$ about $A_0$.
If $(u,u')=(e_i,v_j)$, point $A_2$ lies at the intersection of $s'$ and $\zeta(u,u')$, where $s'$ is the extended line of the 2-scaling of $(e_i\oplus v_j)$ about $A_0$.
To be more clear, we point out the following facts:\\
\quad (i) When $(u,u')=(v_i,e_j)$, line $s$ has at most one intersection with $\zeta(u,u')$.\\
\quad (ii) When $(u,u')=(e_i,v_j)$, line $s'$ has at most one intersection with $\zeta(u,u')$.

\smallskip We only prove (i). The proof of (ii) is symmetric.

Applying part~2 of Lemma~\ref{lemma:dist-unique-location},  $\zeta(v_i,e_j)\subset [v_{j+1}\circlearrowright \D_j]$.
Therefore, for any line that is parallel to $e_j$, such as $s$, it has at most one intersection with $\zeta(v_i,e_j)=\zeta(u,u')$.

\smallskip After $A_2$ is determined, $\M(A_2,A_0)$ is determined and so is $\M(A_1,A_3)$.
Because $A_3\in u$ and $A_1\in u'$, after $\M(A_1,A_3)$ is determined, we can determine $A_1,A_3$ easily in $O(1)$ time due to Claim~\ref{claim:ll'}.
Thus all corners of $Q$ are determined.
 \end{proof}

The framework of our first algorithm is given in Algorithm~\ref{alg:N} below.
It assumes that $A_0$ is an anchored narrow corner and that $A_3\in u$ and $A_1\in u'$
  (so $u$ have to be chasing $u'$, since $A_0$ is narrow).
The correctness of this algorithm follows from Lemma~\ref{lemma:transformed}, which claims that $A_0\in \block(u,u')$ under the assumption.

\begin{algorithm}[h]
\caption{Computing those LMAPs with an anchored narrow corner}\label{alg:N}
\ForEach{$(u,u')$ such that $u$ is chasing $u'$}{
    \ForEach{vertex $V$ in $\block(u,u')$}{
        $A_0\leftarrow V$ and compute $A_2$ as stated in the proof of Claim~\ref{claim:N};\\
        Compute $\M(A_1,A_3)=\M(A_0,A_2)$ and then compute $A_1,A_3$;\\
        Output $Q=A_0A_1A_2A_3$ .
    }
}
\end{algorithm}

To analyze its running time, we need to provide more details -- especially, how do we efficiently find all the vertices in $\block(u,u')$ and compute $A_2$?
(The other steps are easy to analyze by the proof of Claim~\ref{claim:N}.) We need the following lemma.

\begin{lemma}[Monotonicity of the blocks]\label{lemma:block_mono}
Given $u=e_i$. Assume $\D_i=v_k$.
We know $u$ is chasing all units in the list $U'=(e_{i+1},v_{i+2},\ldots,v_{k-1},e_{k-1})$ and no other units.
Therefore, $\block(u,u')$ is defined for $u'\in U'$.
We claim that $$\block(u,e_{i+1}),\block(u,v_{i+2}),\ldots,\block(u,v_{k-1}),\block(u,e_{k-1})$$
 have a monotonicity in the direction perpendicular to $e_i$.
  More formally, when we project these blocks along direction $e_i$ onto some line nonparallel to $e_i$, their images are pairwise-disjoint and in order.
  See Figure~\ref{fig:block_mono} for an illustration.
\end{lemma}

\begin{figure}[h]
\centering\includegraphics[width=.8\textwidth]{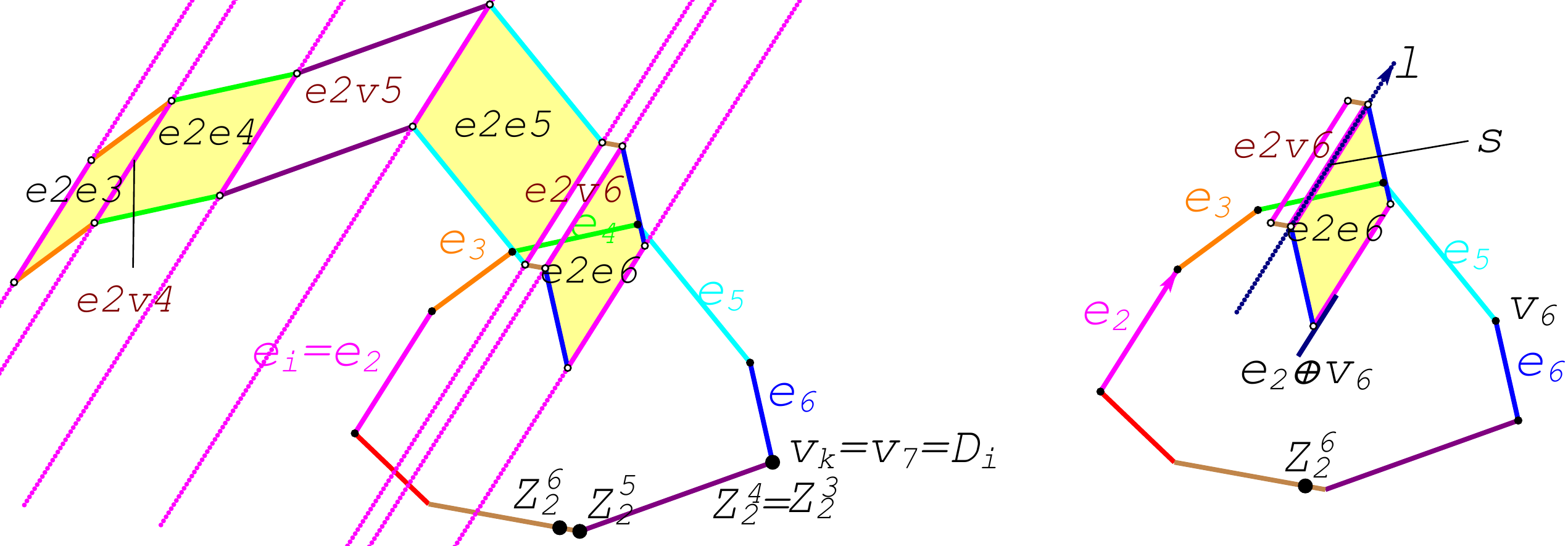}
\caption{Illustration of the monotonicity of the blocks (Lemma~\ref{lemma:block_mono}).}\label{fig:block_mono}
\end{figure}

\begin{proof}
We prove it by the example in Figure~\ref{fig:block_mono}.

First, we argue that $\block(e_2,v_6)$ and $\block(e_2,e_6)$ are separated by a line that is parallel to $e_2$.
Specifically, they are separated by the extended line $l$ of $s$, where $s$ denotes the 2-scaling of $(e_2\oplus v_6)$ about $Z_2^6$ (see the right picture).
Note that $s$ is a translate of $e_2$ and we assume that it has the same direction as $e_2$ (which is from $v_2$ to $v_3$).
By the definition of blocks, $s$ lies in $\block(e_2,v_6)$ and on the boundary of $\block(e_2,e_6)$.
Moreover, $\block(e_2,v_6)$ lies on the left of $s$, because $Z_2^5$ is further than $Z_2^6$ in the distance to the extended line of $e_2$.
  This uses the bi-monotonicity of $Z$-points given in Lemma~\ref{lemma:Z_bi-monotonicity}.
On the contrary, $\block(e_2,e_6)$ lies on the right of $s$, because $e_6$ is further than $v_6$ in the distance to the extended line of $e_2$.

Similarly, $\block(e_2,e_5)$ and $\block(e_2,v_6)$ are separated by the extended line of the 2-scaling of $(e_2\oplus v_6)$ about $Z_2^5$, and so on.
Thus, we obtain this lemma.
 \end{proof}

The details of Algorithm~\ref{alg:N} are given below.

\paragraph{1. How do we enumerate $u,u'$?}
We do the following for every edge $u$ (the case where $u$ is a vertex will be handled later). \emph{When $u=e_i$,
let $U'$ be the same as Lemma~\ref{lemma:block_mono}, and we let $u'$ go through all units in $U'$ in clockwise order}.

\paragraph{2. How do we find every vertex $V$ in $\block(u,u')$?}
  There are two lines parallel to $e_i$, denoted by $l_{u'},r_{u'}$, which separate $\block(u,u')$ from its two neighboring blocks.
\emph{We go through each vertex $V$ that lies between $l_{u'},r_{u'}$ and lies in $[v_i\circlearrowright v_k]$ in clockwise order.}
(Note: all vertices in $\block(u,u')$ will be enumerated in this way, but some enumerated vertices may not be in $\block(u,u')$.
We do not distinguish whether $V\in \block(u,u')$ at this moment since it is unnecessary and too expensive.)

We claim that for a fixed $u$, the total time for enumerating $u'$ and $V$ is only $O(n)$.
To achieve such running time, the main challenge lies in computing the two lines $l_{u'},r_{u'}$ efficiently for every $u'\in U'$.
This reduces to computing the endpoints of $\zeta(u,u')$ for every $u'$, i.e.\ $Z_i^{i+1},\ldots,Z_i^{k-1}$,
  which costs $O(n)$ time due to Lemma~\ref{lemma:Z-compute}.

\paragraph{3. How do we compute $A_2$?} Currently, $u=e_i,u',A_0=V$ are fixed.
  Due to Claim~\ref{claim:N}, we are able to determine $A_2$ now.
  (According to the proof of Claim~\ref{claim:N}, $A_2$ is well-defined even if $A_0\notin \block(u,u')$.)
    However, if we compute $A_2$ according to the proof of Claim~\ref{claim:N} (e.g., if we compute $A_2$ as $Z_i^j$ when $u'=e_j$), it would cost $O(\log n)$ time (by a binary search).
So, we need one more observation here.

\begin{claim}
Fix $u=e_i$, when we go through all $(u',V)$ as above, $A_2$ will always move in clockwise around $\partial P$,
  and can be computed in amortized $O(1)$ time.
\end{claim}

\begin{proof}
We briefly prove this monotonicity of $A_2$ using the example in Figure~\ref{fig:block_mono}.
For $u'=e_5$, no matter where $V$ lies, we know that $A_2$ lies at $Z_2^5$.
Then, let $u'=v_6$. Let $s'_V$ denote the extended line of the 2-scaling of $(e_2\oplus v_6)$ about $V$.
When $V$ moves in clockwise order, $s'_V$ will get more and more closer to the extended line of $e_2$, which means $A_2$,
   the intersection of $s'_{V}$ and $[Z_2^5\circlearrowright Z_2^6]$, moves in clockwise.
   Finally, for $u'=e_6$, point $A_2$ stays at $Z_2^6$.
   To sum up, $A_2$ moves in clockwise in the whole enumerating process of $(u',V)$ for a fixed $u$.

To compute $A_2$, we maintain the unit $w$ in $[v_k\circlearrowright v_i]$ which intersects $s'_V$ and compute $A_2=w\cap s'_V$.
  Notice that $w$ also goes in clockwise order, so we can easily maintain $w$ and compute $A_2$ in amortized $O(1)$ time.
 \end{proof}

\paragraph{4. Two subroutines scheme.} To implement Algorithm~\ref{alg:N}, we use two symmetric subroutines.
Subroutine-1 takes charge of the case where $u$ is an edge, and Subroutine-2 the case where $u'$ is an edge.
By the analysis above, both of them run in $O(n^2)$ time.
However, we should modify these subroutines so that they can handle the case where both $u,u'$ are vertices.
We show this modification below.

\begin{figure}[h]
\centering\includegraphics[width=.43\textwidth]{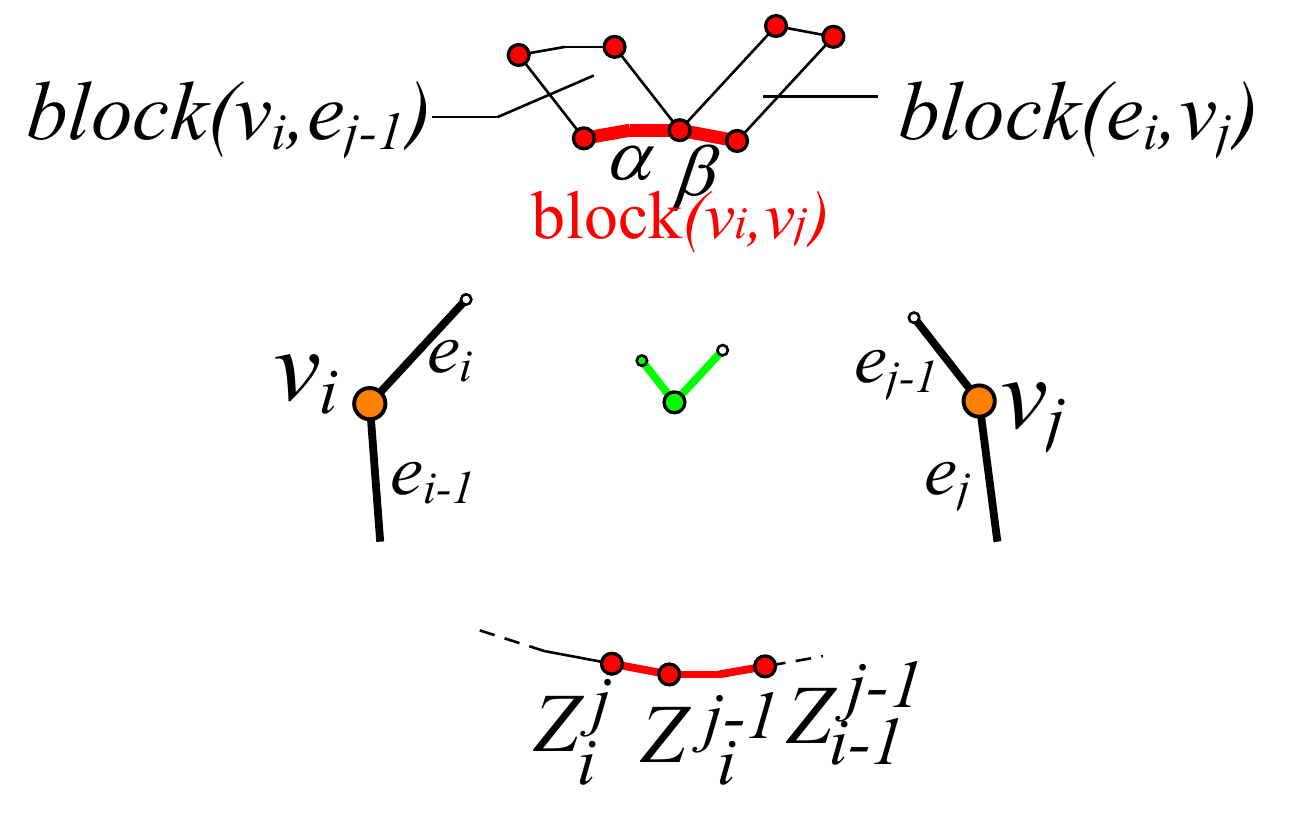}
\caption{$\block(v_i,v_j)$ is the concatenation of two curves $\alpha,\beta$, where $\alpha$ is the ``bottom border'' of $\block(v_i,e_{j-1})$, and
 $\beta$ is the ``bottom border'' of $\block(e_i,v_j)$.}\label{fig:alg-N-VV}
\end{figure}

\begin{claim}\label{claim:border}
     Assume $v_i$ is chasing $v_j$ ($j\neq i+1$). See Figure~\ref{fig:alg-N-VV}.
     Then, $\block(v_i,v_j)$ is the concatenation of $\alpha$ and $\beta$,
       where $\alpha$ is the reflection of $\zeta(v_i,e_{j-1})=[Z_{i-1}^{j-1}\circlearrowright Z_i^{j-1}]$ around $\M(v_i,v_j)$
         -- the bottom border of $\block(v_i,e_{j-1})$,
       and $\beta$ is the reflection of $\zeta(e_i,v_j)=[Z_i^{j-1} \circlearrowright Z_i^j]$ around $\M(v_i,v_j)$
         -- the bottom border of $\block(e_i,v_j)$.
\end{claim}
\begin{proof}
 Using the bi-monotonicity of $Z$-points, $Z_i^{j-1}$ lies in $[Z_{i-1}^{j-1}\circlearrowright Z_i^j]$.
  \end{proof}

\paragraph{5. Modification of subroutines.}
Suppose we are under the case $(u,u')=(e_i,v_j)$ in Subroutine-1 and are enumerating all the vertices between $l_{u'},r_{u'}$.
As mentioned, we may find some vertex $V$ not in $\block(e_i,v_j)$.
As a result, after the computation of $A_2$ and $(A_1,A_3)$, we may find that $A_3$ is not contained in $e_i$.
Moreover, if $V$ lies in the ``bottom border'' of $\block(e_i,v_j)$ (namely, curve $\beta$), we will get that $A_3$ lies in $v_i$.
This can be seen from Figure~\ref{fig:alg-N-VV}.
When this case happens, we are not taking care of $(u,u',V)=(e_i,v_j,V)$, but actually taking care of a combination $(u,u',V)=(v_i,v_j,V)$.
Similarly, in Subroutine-2, when we find a vertex $V$ that lies in the ``bottom border'' of $\block(v_i,e_{j-1})$ (namely, curve $\alpha$),
  we are actually taking care of a combination $(v_i,v_j,V)$.
On the contrary, by Claim~\ref{claim:border}, every combination $(v_i,v_j,V)$ will be taken care of, either in Subroutine-1 or in Subroutine-2.

\subsection{Compute the LMAPs with an anchored broad corner which has an adjacent corner anchored}

Assume now $Q=A_0A_1A_2A_3$ is an LMAP with an anchored broad corner $A_2$ which has at least one adjacent corner anchored.
Assume $A_3\in u$ and $A_1\in u'$, where $(u,u')$ is a pair of units such that $u$ is chasing $u'$, and $u,u'$ are not both edges.

\begin{claim}\label{claim:B}
Under the above assumption, when $u,u'$ and the position of $A_2$ are given, the positions of four corners of $Q$ are all determined.
\end{claim}

\begin{proof}
We first determine $A_0$. Consider three cases.
If $(u,u')=(v_i,v_j)$, point $A_0$ lies at the reflection of $A_2$ around $\M(v_i,v_j)$.
If $(u,u')=(v_i,e_j)$, point $A_0$ lies at the intersection of $s$ and $[v_i\circlearrowright v_j]$,
  where $s$ is the extended line of the 2-scaling of $(v_i\oplus e_j)$ about $A_2$.
If $(u,u')=(e_i,v_j)$, point $A_0$ lies at the intersection of $s'$ and $[v_{i+1}\circlearrowright v_j]$,
  where $s'$ is the extended line of the 2-scaling of $(e_i\oplus v_j)$ about $A_2$.

After $A_0$ is determined, $\M(A_0,A_2)=\M(A_1,A_3)$ is determined, then $A_1,A_3$ are determined in $O(1)$ time;
  see the proof of Claim~\ref{claim:N} in the previous subsection.
 \end{proof}

\begin{algorithm}[h]
\caption{Computing those LMAPs with an anchored broad corner which has at least one adjacent corner anchored}\label{alg:B}
\ForEach{$(u,u')$ such that $u$ is chasing $u'$ and $u,u'$ are not both edges}{
    \ForEach{vertex $V$ in $\zeta(u,u')$}{
        $A_2\leftarrow V$ and compute $A_0$ as stated in the proof of Claim~\ref{claim:B};\\
        Compute $\M(A_1,A_3)=\M(A_0,A_2)$ and then compute $A_1,A_3$;\\
        Output $Q=A_0A_1A_2A_3$ .
    }
}
\end{algorithm}

The framework of our second algorithm is given in Algorithm~\ref{alg:B}.
Generally, it tries all the combinations of $(u,u',A_2)$ and compute the potential parallelogram $Q=A_0A_1A_2A_3$ using Claim~\ref{claim:B}.
The correctness follows from the clamping bounds (Lemma~\ref{lemma:clamping-broad}) which says $A_2\in \zeta(u,u')$.
More details are given below.

We implement Algorithm~\ref{alg:B} by four subroutines.
Subroutine-1 handles the case where $u$ is an edge and Subroutine-2 the case where $u'$ is an edge, and
two other subroutines are shown in the next paragraph.
In Subroutine-1, we first let $u$ go through each edge $e_i$ and $u'$ go through each vertex $v_j$ in clockwise order.
For a fixed pair $(u,u')=(e_i,v_j)$, let $V$ go through each vertex in $\zeta(e_i,v_j)=[Z_i^{j-1}\circlearrowright Z_i^j]$ in clockwise order.
Applying the bi-monotonicity of $Z$-points (Lemma~\ref{lemma:Z_bi-monotonicity}), the total time for enumerating $(u',V)$ is only $O(n)$ for fixed $u$.
Moreover, when we enumerate $(u',V)$ in this way, $A_0$ also goes in clockwise, and thus can be computed in amortized $O(1)$ time.
  This is symmetric to ``How do we compute $A_2$?'' in the previous subsection. Subroutine-2 is symmetric to Subroutine~1.

\smallskip We take care of the case $(u,u')=(v_i,v_j)$ by Subroutine-3 and Subroutine-4.
In Subroutine-3, let $u$ go through every vertex $v_i$ and $u'$ go through every vertex $v_j$ in clockwise order,
  and then let $V$ go through every vertex in $\zeta(v_i,e_{j-1})$ (\textbf{not} in $\zeta(v_i,v_j)$).
  In this way, we make sure that $V$ will move in clockwise for a fixed $u$, so $V$ and $A_0$ can be computed in amortized $O(1)$ time.
In Subroutine-4, let $u'$ go through every vertex $v_j$ and $u$ go through every vertex $v_i$ in clockwise order,
  and then let $V$ go through every vertex $V$ in $\zeta(e_i,v_j)$ (\textbf{not} in $\zeta(v_i,v_j)$).
  Similarly, $V$ will move in clockwise for a fixed $u'$, so $V$ and $A_0$ can be computed in amortized $O(1)$ time.
Together, because $\zeta(v_i,v_j)$ is the concatenation of $\zeta(v_i,e_{j-1})$ and $\zeta(e_i,v_j)$,
  every combination $(v_i,v_j,V)$ such that $V\in\zeta(v_i,v_j)$ will be considered.

\subsection{Compute the LMAPs with four even corners}

\newcommand{\SH}{\mathsf{H}}

Every edge $e_i$ is \emph{incident} to two vertices $v_i$ and $v_{i+1}$,
  and every vertex $v_i$ is \emph{incident} to two edges $e_{i-1}$ and $e_{i}$. (Subscripts taken modulo $n$.)
Two units are \emph{non-incident} if they are not incident to each other. (In particular, $e_i,e_{i+1}$ are non-incident.)
For each vertex $V=v_i$, denote by $\SH_V$ the set of units which lie in $(\D_{i-1}\circlearrowright \D_i)$.

\begin{claim}\label{claim:H}
  Assume $u,u'$ are distinct and non-incident units and neither of them is chasing the other. Then, one of the following holds:\\
    (a) $u$ is some vertex $V$, and $u'\in \SH_V$. \\ (b) $u'$ is some vertex $V$, and $u\in \SH_V$.
\end{claim}

\begin{proof} Notice that at least one of $u,u'$ is a vertex; otherwise one of them is chasing the other. So we have the following three cases.
\begin{enumerate}
\item[Case~1:] $u$ is a vertex and $u'$ is an edge. We claim $u'\in \SH_{u}$ and so (a) holds.\\
    Assume $u=v_j$. First, since $u,u'$ are non-incident, $u'\notin \{e_j,e_{j-1}\}$.
    Second, because $u$ is not chasing $u'$, edge $u'$ is not contained in $[ v_{j+1}\circlearrowright \D_{j-1}]$.
    Third, because $u'$ is not chasing $u$, edge $u'$ is not contained in $[\D_j\circlearrowright v_{j-1}]$.
    Therefore, edge $u'$ can only lie in $(\D_{j-1}\circlearrowright \D_j)$, namely, $u'\in\SH_{u}$.

\item[Case~2:] $u$ is an edge and $u'$ is a vertex. Symmetrically, $u\in \SH_{u'}$ and (b) holds.

\item[Case~3:] $u,u'$ are both vertices. Assume that $u=v_j,u'=v_k$.\\
    First, consider the case where $e_j\prec e_k$.
    Then, $e_{k-1}\prec e_{j-1}$, otherwise $v_j$ is chasing $v_k$.
    Since $e_j\prec e_k$, we get $v_k\in (v_j\circlearrowright \D_j)$. Since $e_{k-1}\prec e_{j-1}$, we get $v_k\in (\D_{j-1}\circlearrowright v_j)$.
    Together, $v_k\in (\D_{j-1}\circlearrowright \D_j)$, i.e. $u'\in \SH_{u}$.\\
    For the other case where $e_k\prec e_j$, we can get $u\in \SH_{u'}$ symmetrically.
\end{enumerate}
\end{proof}

Denote $S=\{(u,u')\mid \text{$u$ is a vertex and $u'\in \SH_u$, or $u'$ is a vertex and $u\in \SH_{u'}$}\}.$

Assume $Q=A_0A_1A_2A_3$ is an LMAP with four even corners.
Denote $u_0=\unit(A_0),u_1=\unit(A_1),u_2=\unit(A_2),u_3=\unit(A_3)$.
Obviously, units $u_0,u_2$ are distinct, non-incident, and neither of them is chasing the other (because $A_1,A_3$ are even corners), so $(u_0,u_2)\in S$ due to Claim~\ref{claim:H}.
Similarly, $(u_1,u_3)\in S$.
Therefore, to compute those LMAPs with four even corners, we only need to try all possible choices of $(u_0,u_2)$ and $(u_1,u_3)$.
The algorithm is presented in Algorithm~\ref{alg:EE}.

\begin{algorithm}[h]
\caption{Computing those LMAPs with four even corners.}\label{alg:EE}
\ForEach{$(u_0,u_2)$ in $S$}{
    \ForEach{$(u_1,u_3)$ in $S$}{
        Let $C$ be the intersection between $(u_0\oplus u_2)$ and $(u_1\oplus u_3)$.\\
        Compute $A_0,A_2$ from $u_0,u_2,C$ and compute $A_1,A_3$ from $u_1,u_3,C$;\\
        Output $Q=A_0A_1A_2A_3$.
    }
}
\end{algorithm}

By the definition of $\SH$, set $S$ contains only $O(n)$ pairs of units.
Therefore, there are $O(n^2)$ choices for $(u_0,u_2,u_1,u_3)$.
Moreover, after $u_0,u_1,u_2,u_3$ are determined, it is easy to compute $A_0,A_1,A_2,A_3$ in $O(1)$ time,
  so the running time is $O(n^2)$.
  
\smallskip As a summary of this section, we have proved the following result.
\begin{theorem}[Main result]
Given a convex polygon $P$ bounded by $n$ halfplanes, we can compute the LMAPs and MAPs in $P$ in $O(n^2)$ time.
\end{theorem}

\section{A monotonicity property of the LMAPs in a convex polygon}\label{sect:interleave}

This section proves an interesting monotonicity property of the LMAPs as stated in the next theorem.
  The proof applies once again of our basic tools given in Section~\ref{sect:identity}.
    (Recall that such tools were applied in proving the clamping bounds in Section~\ref{sect:clamping}.)

\begin{definition}[Interleaving]\label{def:interleaving}
Assume $A=A_0\ldots A_{k-1}$ and $B=B_0\ldots B_{k-1}$ are two $k$-gons inscribed in $P$,
  where the corners $A_0,\ldots,A_{k-1}$ lie in clockwise order and so do $B_0,\ldots,B_{k-1}$.
We say $A,B$ \emph{interleave} if
\begin{enumerate}
\item for any two neighboring corners $A_i,A_{i+1}$ of $A$, we can find a corner of $B$ which lies in $[A_i\circlearrowright A_{i+1}]$, and
\item for any two neighboring corners $B_i,B_{i+1}$ of $B$, we can find a corner of $A$ which lies in $[B_i\circlearrowright B_{i+1}]$ (subscripts are taken modulo $k$).
\end{enumerate}
\end{definition}

\begin{figure}[h]
\centering\includegraphics[width=.85\textwidth]{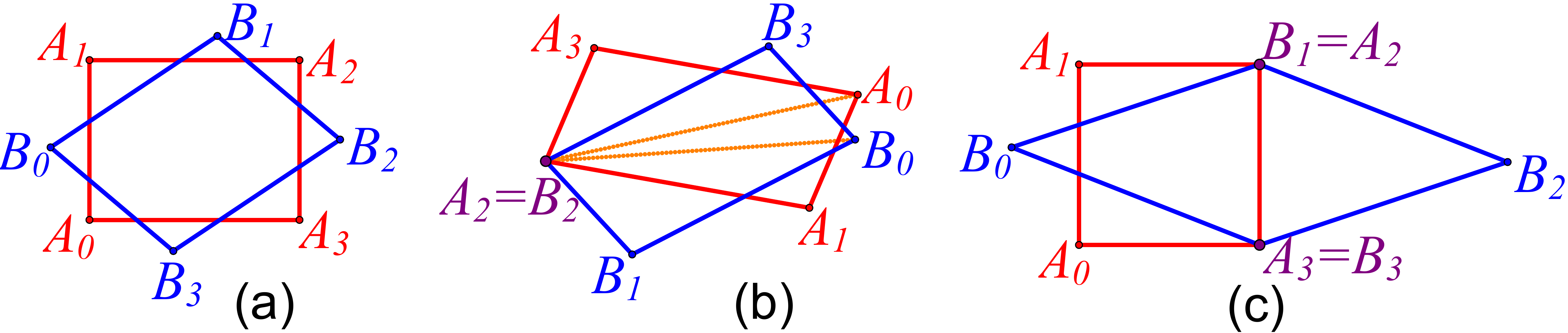}
\caption{Three examples in which $A$ interleaves $B$.}\label{fig:interleave-examples}
\end{figure}

\begin{theorem}[Auxiliary]\label{thm:interleave}
Any two LMAPs in the convex polygon $P$ interleave.
\end{theorem}

\begin{remark}\label{remark:interleave}
It is well-known that all the locally maximal triangles in $P$ interleave each other.
Recently, by utilizing this interleaving property,
  Jin \cite{Triangle-ultimate-Arxiv} gave a linear time algorithm for computing all the locally maximal triangles.
Hence Theorem~\ref{thm:interleave} might be useful in designing a better algorithm for finding the LMAPs in the future.
\end{remark}

We prove one preliminary lemma before proving Theorem~\ref{thm:interleave}.

\begin{lemma}\label{lemma:inter-pre}
See Figure~\ref{fig:interleave0}~(a).
Assume that $A=A_0A_1A_2A_3$ and $B=B_0B_1B_2B_3$ are parallelograms inscribed in $P$ with the following properties:
\begin{enumerate}
\item[(a)] $A_0,A_2,B_2,B_0$ are distinct and lie in clockwise order.
\item[(b)] $A_1,B_1$ both lie in $\rho=(A_0\circlearrowright A_2)$, and $A_3,B_3$ both lie in $\rho'=(B_2\circlearrowright B_0)$.
\end{enumerate}
We claim that the following holds:
\begin{enumerate}
\item $A_1\neq B_1$ and symmetrically, $A_3 \neq B_3$.
\item When $B_1<_\rho A_1$, then $B_3<_{\rho'}A_3$, and when $A_1<_\rho B_1$, then $A_3<_{\rho'} B_3$.
\end{enumerate}
\end{lemma}

\begin{figure}[h]
\centering\includegraphics[width=.82\textwidth]{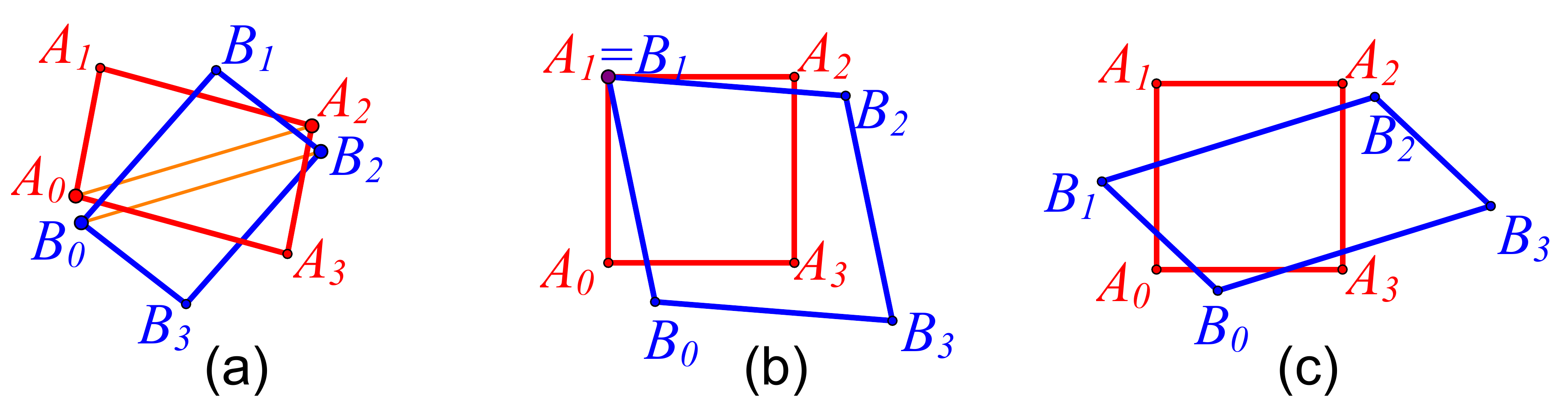}
\caption{Proof of the interleaving property: - a preliminary lemma.}\label{fig:interleave0}
\end{figure}

\begin{proof}[Proof of Lemma~\ref{lemma:inter-pre}]
Without loss of generality, assume $A$ is an axis-parallel unit square with $A_0=(0,0)$ and $A_2=(1,1)$.
Otherwise we use a linear transformation to make it so - the parallelograms remain parallelograms after any non-degenerate linear transformations, and $P$ remains a convex polygon.

Let $X.x$ and $X.y$ denote the x-coordinate and y-coordinate of any point $X$.

See Figure~\ref{fig:interleave0}~(a) and (b). By the assumptions, $A_2,B_2,A_3,B_0,A_0$ are distinct and lie in clockwise order.
Thus, $B_2.x > 1$, $0<B_2.y<1$, $0<B_0.x<1$, $B_0.y<0$.

First, we prove that $B_1\neq A_1$.
  Suppose to the opposite that $B_1=A_1$, as shown in Figure~\ref{fig:interleave0}~(b).
  Let $C_B$ denote the center of $B$.
  We know
  \begin{equation*}
    \left\{
    \begin{gathered}
       B_3.x=2C_B.x-B_1.x=B_0.x+B_2.x-0>1.\\
       B_3.y=2C_B.y-B_1.y=B_0.y+B_2.y-1<0.
    \end{gathered}\right.
  \end{equation*}
  This means $A_3$ is contained in the interior of $\triangle A_0A_2B_3$. So, $A_3$ is not on the boundary of $P$. Contradiction!
    Thus $B_1\neq A_1$. Symmetrically, $B_3\neq A_3$.\smallskip

  Next, assume that $B_1<_\rho A_1$ and we prove that $B_3<_{\rho'}A_3$. See Figure~\ref{fig:interleave0}~(c).
    Similar as above, we have $B_3.x>1$. Hence $0<B_3.y<1$ (otherwise $P$ is not convex).
      Therefore, $B_3<_{\rho'}A_3$.
      Symmetrically, when $A_1<_\rho B_1$, we can get $A_3<_{\rho'}B_3$.
 \end{proof}

\begin{proof}[Proof of Theorem~\ref{thm:interleave}]
Let $A=A_0A_1A_2A_3$ and $B=B_0B_1B_2B_3$ be two LMAPs.

\paragraph{Step~1.} First, we show that \emph{any diagonal of $B$ cannot lie outside $A$}.
To be more rigorous, the following three cases are forbidden.\\
(I)  $B_i,B_{i+2}$ lies in $(A_j\circlearrowright A_{j+1})$ for some $i,j$ (see Figure~\ref{fig:interleave1}~(a)).\\
(II) $B_i=A_j$ and $B_{i+2}\in (A_j\circlearrowright A_{j+1})$ for some $i,j$ (see Figure~\ref{fig:interleave1}~(b)).\\
(III) $B_{i+2}=A_{j+1}$ and $B_i \in (A_j\circlearrowright A_{j+1})$ for some $i,j$ (symmetric to (II)).\\
(But the diagonal of $B$ may coincide with an edge of $A$, as shown in Figure~\ref{fig:interleave-examples}~(c).)

\begin{figure}[h]
\centering\includegraphics[width=.77\textwidth]{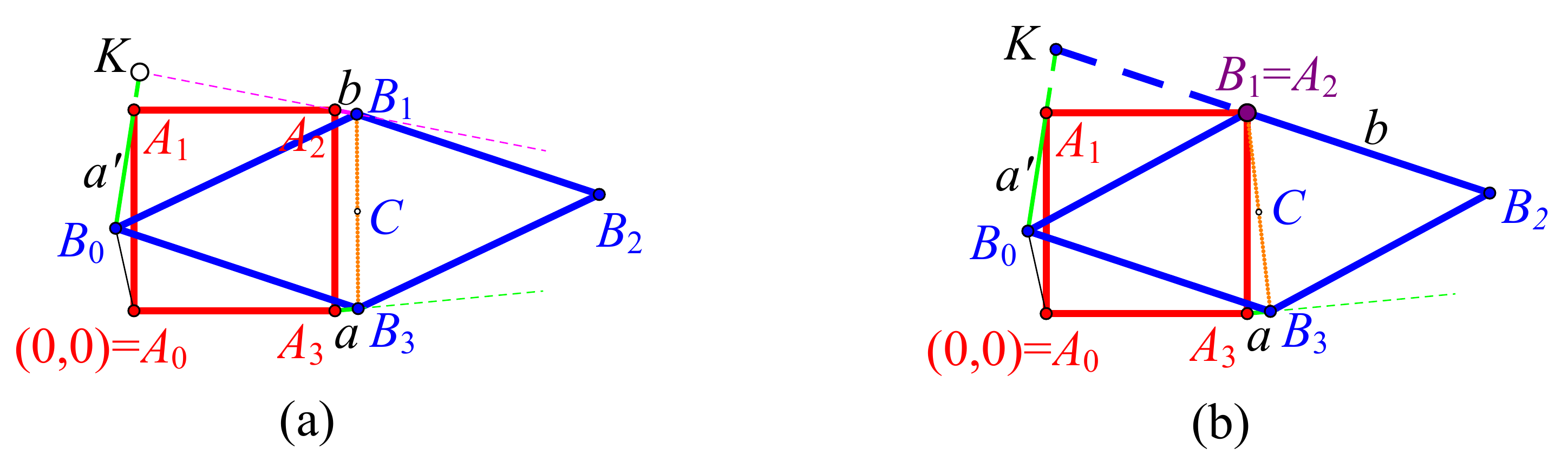}
\caption{Proof of the interleaving property: any diagonal of $B$ cannot lie outside $A$.}\label{fig:interleave1}
\end{figure}

We only show the proof of (I). The others are similar.

For a contradiction, suppose that $B_1,B_3$ lie in $(A_2\circlearrowright A_3)$, as shown in Figure~\ref{fig:interleave1}~(a).
Without loss of generality, assume $A$ is an axis-parallel unit square with $A_0=(0,0)$ and $A_2=(1,1)$.
Otherwise we use linear transformations - the LMAPs remain LMAPs after any non-degenerate linear transformations.

Let $a=\overline{A_3B_3},a'=\overline{A_1B_0},b=\overline{A_2B_1}$.
Let $l_a,l_{a'}$ denote the extended lines of $a,a'$ respectively.
Assume they intersect the extended line of $b$ at $J,K$.
Let point $X$ be restricted to $b$, and denote $Q_X=\parallelogram(A_0,X,l_a,l_{a'})$. We state two observations:
(i) When $X$ is sufficiently close to $A_2$, the two corners of $Q_X$ that are restricted to the extended lines of $a,a'$ will lie in segments $a,a'$,
and so $Q_X\in P$.
(ii) $Area(Q_X)$ increases when $X$ moves from $A_2$ towards $B_1$ for a sufficiently small distance.

Clearly, (i) and (ii) together imply that $Q_{A_2}=A$ is not locally maximal.

\smallskip \noindent \emph{Proof of (i).} When $X$ goes from $A_2$ to $B_1$, the center of $Q_X$ gets closer to $l_a$,
  thus the corner of $Q_X$ restricted to $l_{a'}$ moves toward $B_0$ and will stay in $a'$ if $X$ is sufficiently close to $A_2$.
  The center of $Q_X$ gets away from $l_{a'}$, thus the corner of $Q_X$ restricted to $l_{a}$ moves toward $B_3$ and will stay in $a$ if $X$ is sufficiently close to $A_2$.

\medskip \noindent \emph{Proof of (ii).} By Lemmas~\ref{lemma:area-pre} and \ref{lemma:area},
  $Area(Q_X)$ is proportional to $\dist_{l_a,l_{a'}}(X)-\dist_{l_a,l_{a'}}(A_0)$.
Further applying the concavity of $\dist_{l_a,l_{a'}}$ (Lemma~\ref{lemma:dist_concave}),
  proving (ii) reduces to proving that $A_2$ lies strictly between $K$ and $\M(J,K)$.
  Equivalently, we should prove that the x-coordinate of $\M(J,K)$ is larger than 1.

  Let $C$ be the center of $B$.
  Observing the figure, from the assumption we have:
  $B_0.x\leq 0,C.x>1$. Together, $B_2.x>2$. Further, observe that $J.x$ must be larger than $B_2.x$, so $J.x>2$.
  Moreover, we have $K.x>0$. Together, $\M(J,K).x>1$.

\paragraph{Step~2.} Second, we show that \emph{any diagonal of $B$ intersects any diagonal of $A$}.
Rigorously, we show that any endpoint-inclusive segment $\overline{B_iB_{i+2}}$ shares at least one common point with any endpoint-inclusive segment $\overline{A_jA_{j+2}}$ for any $i,j$.
For example, in Figure~\ref{fig:interleave-examples}~(b), the two diagonals $A_0A_2$ and $B_0B_2$ intersect at $A_2=B_2$.

\smallskip For a contradiction, assume that $A_0A_2$ does not intersect $B_0B_2$ and that $A_0,A_2,B_2,B_0$ lie in clockwise order.
Denote $\rho=(A_0\circlearrowright A_2)$ and $\rho'=(B_2\circlearrowright B_0)$.
Clearly, $A_1\in \rho$ and $B_3\in \rho'$.
Moreover, applying the result of Step~1, $B_1\in \rho$ and $A_3\in\rho'$.
Therefore, applying Lemma~\ref{lemma:inter-pre}, there are two possibilities:
$$\text{1. $B_1<_\rho A_1$ and $B_3<_{\rho'}A_3$, or, 2. $A_1<_\rho B_1$ and $A_3<_{\rho'} B_3$.}$$

\begin{figure}[h]
\centering\includegraphics[width=.66\textwidth]{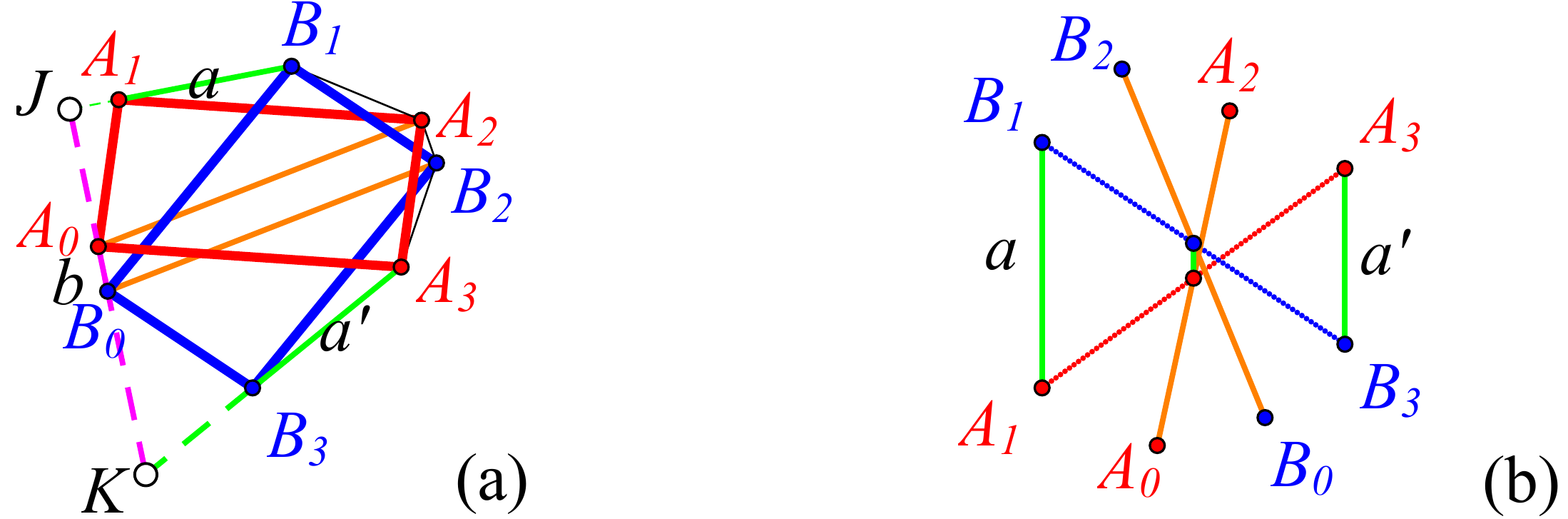}
\caption{Proof of the interleaving property: - diagonals from $A,B$ must intersect.}\label{fig:interleave2}
\end{figure}

In the next, assume that $A_1<_\rho B_1$ and $A_3<_{\rho'} B_3$, as shown in Figure~\ref{fig:interleave2}~(a).

\smallskip Let $a=\overline{A_1B_1},a'=\overline{A_3B_3}$.
We state that \emph{$a$ is not parallel to $a'$}. The proof is deferred for a moment.
Let $l_a,l_{a'}$ be the extended lines of $a,a'$ respectively.
Without loss of generality, assume $l_a,l_{a'}$ intersect on rays $\overrightarrow{A_1B_1}$ and $\overrightarrow{B_3A_3}$; otherwise it is symmetric.
Let $b=\overline{A_0B_0}$ and assume its extended line intersects $l_{a},l_{a'}$ at points $J,K$ respectively.
Assume $X$ lies in segment $b$. Denote $Q^A_X=\parallelogram(X,A_2,l_a,l_{a'})$ and $Q^B_X=\parallelogram(X,B_2,l_a,l_{a'})$ for short.
We state four observations:
\begin{enumerate}
\item[(i')] If $X$ is sufficiently close to $A_0$, parallelogram $Q^A_X$ is inscribed in $P$.
\item[(ii')] If $|JA_0|<\frac{1}{2}|JK|$, $Area(Q^A_X)$ grows when $X$ moves from $A_0$ to $\M(J,K)$.
\item[(i'')] If $X$ is sufficiently close to $B_0$, parallelogram $Q^B_X$ is inscribed in $P$.
\item[(ii'')] If $|KB_0|<\frac{1}{2}|JK|$, $Area(Q^B_X)$ grows when $X$ moves from $B_0$ to $\M(J,K)$.
\end{enumerate}
The proofs of these observations are similar to those of (i) and (ii) and omitted.

\smallskip Since $|JA_0|<\frac{1}{2}|JK|$ or $|KB_0|<\frac{1}{2}|JK|$,
applying the above observations, $A=Q^A_{A_0}$ is not locally maximal or $B=Q^B_{B_0}$ is not locally maximal.

\medskip Now we prove that $a\nparallel a'$.
Suppose to the opposite that $a$ is parallel to $a'$, as shown in Figure~\ref{fig:interleave2}~(b).
Assume $l_a,l_{a'}$ are vertical lines for convenience.
Since $B_3,B_0,A_0,A_1$ lie in clockwise order (observe Figure~\ref{fig:interleave2}~(a) here), we know $(A_0).x < (B_0).x$.
Further since $(A_0).x+(A_2).x=(B_0).x+(B_2).x$, we have $(A_2).x > (B_2).x$, which means $B_1,B_2,A_2,A_3$ lie in clockwise order,
  which contradicts the fact that $B_1,A_2,B_2,A_3$ lie in clockwise order.

\paragraph{Step~3.} Finally, we prove that $A$ interleaves $B$ by contradiction.
Suppose without loss of generality that no corner of $B$ lies in $[A_0\circlearrowright A_1]$.

\smallskip First, we prove the following facts.\\
(a) No corner of $B$ can lie in $(A_2\circlearrowright A_3)$.\\
(b) No corner of $B$ can lie at $A_2$ or $A_3$.

\smallskip \noindent \emph{Proof of (a).} We prove it by contradiction.
    Assume without loss of generality that $B_2$ lies in $(A_2\circlearrowright A_3)$.
By assumption, $B_0\notin [A_0\circlearrowright A_1]$. So segment $B_0B_2$ cannot intersect $A_0A_2$ and $A_1A_3$ simultaneously.
This contradicts the result of Step~2.

\smallskip \noindent \emph{Proof of (b).} We prove it by contradiction.
    Assume some corner of $B$ lies at $A_2$. (Similar arguments hold for $A_3$.)
Without loss of generality, assume $B_2=A_2$.

Consider the position of $B_0$. According to the result of Step~1, $B_0\notin (A_1\circlearrowright A_3)$.
Moreover, it cannot lie in $[A_0\circlearrowright A_1]$ (there is no such corner in $B$ by assumption).
Furthermore, it cannot lie at $A_3$, since otherwise $B_3$ can only lie in $(A_2\circlearrowright A_3)$ which contradicts (a).
Therefore, $B_0$ can only lie in $(A_3\circlearrowright A_0)$.

Next, consider the position of $B_3$. Assume that $B_3\neq A_3$ (otherwise, parallelograms $A,B$ share the same side $A_2A_3=B_2B_3$
    and it must happen that $A_0,A_1,B_0,B_1$ lie in the same line and thus $A$ interleaves $B$).
  Moreover, by fact (a), $B_3\notin (A_2\circlearrowright A_3)$.
  Therefore, $B_3$ also lies in $(A_3\circlearrowright A_0)$.

Now, the parallelograms $A,B$ are illustrated in Figure~\ref{fig:interleave3}~(a).

\begin{figure}[h]
\centering\includegraphics[width=\textwidth]{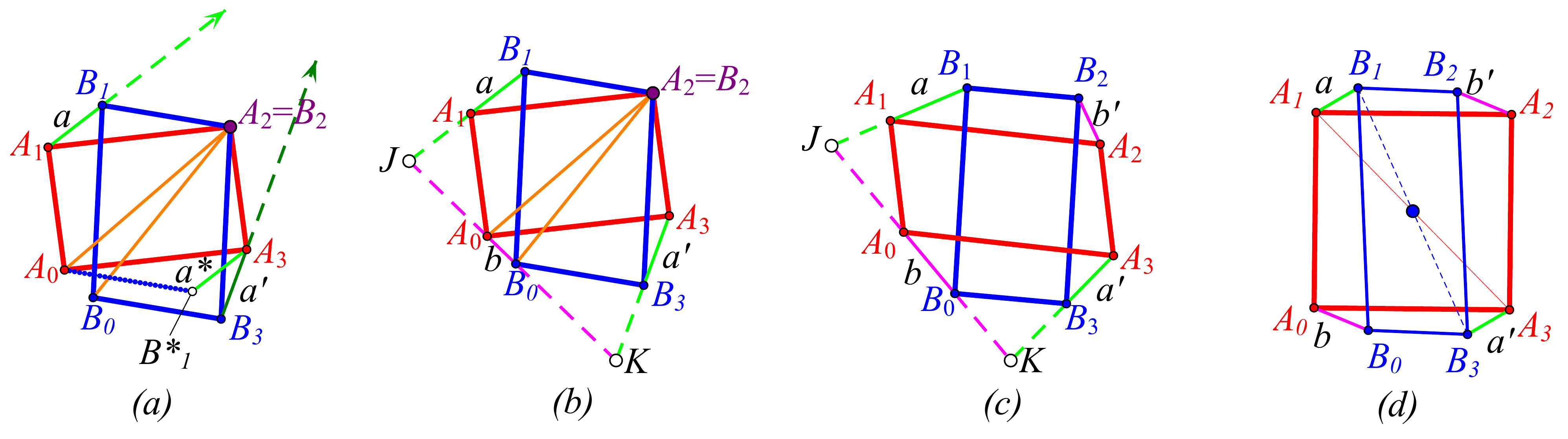}
\caption{Proof of the interleaving property of LMAPs - last step.}\label{fig:interleave3}
\end{figure}

\smallskip We claim that (X) \emph{the extended lines of $\overline{A_1B_1},\overline{B_3A_3}$ intersect in rays $\overrightarrow{A_1B_1}$ and $\overrightarrow{B_3A_3}$}.
See Figure~\ref{fig:interleave3}~(a). Make the symmetric point $B^*_1$ of $B_1$ around the center of $A$.
Clearly, $A_0B^*_1$ is a translate of $B_1A_2$ and so a translate a $B_0B_3$.
So, $B^*_1$ is on the right of $\overrightarrow{B_3B_0}$.
Similarly, $B^*_1$ is on the right of $\overrightarrow{B_2B_3}$ since $A_0$ is on the right of $\overrightarrow{B_1B_0}$.
Together, $B^*_1$ is on the right of $\overrightarrow{A_3B_3}$. Further since $A_1B_1\parallel A_3B^*_1$, we get (X).
Using (X), we can deduce that one of $A,B$ is not locally maximal.
The proof is the same as the proof in Step~2 (compare Figure~\ref{fig:interleave3}~(b) with Figure~\ref{fig:interleave2}~(a)) and thus is omitted.
This means $B_2$ cannot lie at $A_2$. Thus fact~(b) holds.

\smallskip Combing (a) and (b), no corner of $B$ lies in $[A_2\circlearrowright A_3]$.
Recall that we also know no corner of $B$ lies in $[A_0\circlearrowright A_1]$.
Therefore, it must be the case that two corners of $B$ lie in $(A_1\circlearrowright A_2)$
 whereas another two lie in $(A_3\circlearrowright A_0)$, as shown in Figure~\ref{fig:interleave3}~(c).

\medskip Let $a=\overline{A_1B_1}$, $a'=\overline{A_3B_3}$, $b=\overline{A_0B_0}$, $b'=\overline{A_2B_2}$.
We shall prove that one of $A,B$ is not locally maximal.
Using the technique developed in Step~2, this is clear when $a\nparallel a'$ or $b\nparallel b'$.
So, assume that $a\parallel a'$ and $b\parallel b'$ in the next. See Figure~\ref{fig:interleave3}~(d).

\smallskip Clearly, $a'$ is a translate of $a$ and $b'$ a translate of $b$. So the centers of $A,B$ coincide.
If $A,B$ are both locally maximal, we have:\\
(1) $A_2$ is further than $B_2$ to the extended line of $A_1A_3$.\\
(2) $B_2$ is further than $A_2$ to the extended line of $B_1B_3$.

Clearly, these two facts contradict. So one of $A,B$ is not locally maximal.
\end{proof}

\begin{corollary}\label{corol:count}
For any given convex polygon $P$, there are only $O(n)$ LMAPs.
\end{corollary}

\begin{proof}
Assume there are $m$ LMAPs. Denote all of them by $\{A^{(i)}_0A^{(i)}_1A^{(i)}_2A^{(i)}_3\mid 1\leq i \leq m\}$.
By Theorem~\ref{thm:interleave}, we can assume that
$$A^{(1)}_1, \ldots, A^{(m)}_1,A^{(1)}_2, \ldots, A^{(m)}_2,A^{(1)}_3, \ldots, A^{(m)}_3, A^{(1)}_4, \ldots, A^{(m)}_4$$
are in clockwise order (where coincidences are allowed).

Now, let $\delta^{(i)}_j$ denote the number of units in $\{\unit(A^{(1)}_j),\ldots,\unit(A^{(i)}_j)\}$, and $\delta^{(i)}=\delta^{(i)}_1+\ldots+\delta^{(i)}_4$.
Since there are only $2n$ units, $4\leq \delta^{(1)}<\ldots<\delta^{(m)}< 2n+4$. So $m\leq 2n$.
The inequality $\delta^{(i)}<\delta^{(i+1)}$ follows from a fact that there exists $j\in \{0,1,2,3\}$ such that $\unit(A^{(i)}_j)\neq \unit(A^{(i+1)}_j)$;
i.e., the four units containing the four corners cannot be all the same for different LMAPs. (Further details are omitted;
    this corollary is proved more clearly in our follow-up work \cite{followup-MAP-arxiv15}.)
 \end{proof}

\section*{Acknowledgements}
We thank professors Haitao Wang and Kevin Matulef for taking part in discussions and for their unselfish helps,
  and thank Andrew C. Yao, Jian Li, Danny Chen, Matias Korman, Wolfgang Mulzer, and Donald Sheehy for their precious suggestions.
Last but not least, we appreciate the developers of Geometer's Sketchpad${}^\circledR$.

\appendix

\section{Literature of Theorem~\ref{theorem:approx}}\label{sect:ratio}

Theorem~\ref{theorem:approx} follows a combination of the following two results.
The first result \cite{Math-Sas-CM39} says that for convex bodies in ${\mathbb R}^d$, the hardest to approximate with inscribed $n$-gons are exactly the ellipsoids.
The second result \cite{Math-Dow-AMS44} says that for any centrally-symmetric convex body $K$ in the plane, and any even $n \ge 4$,
among the inscribed (or contained) convex $n$-gons of maximal area in $K$, there is one which is centrally-symmetric.

\section{A proof of Lemma~\ref{lemma:local-maximal-non-slidable}}\label{sect:omit-slidable}

\begin{figure}[h]
  \centering \includegraphics[width=.53\textwidth]{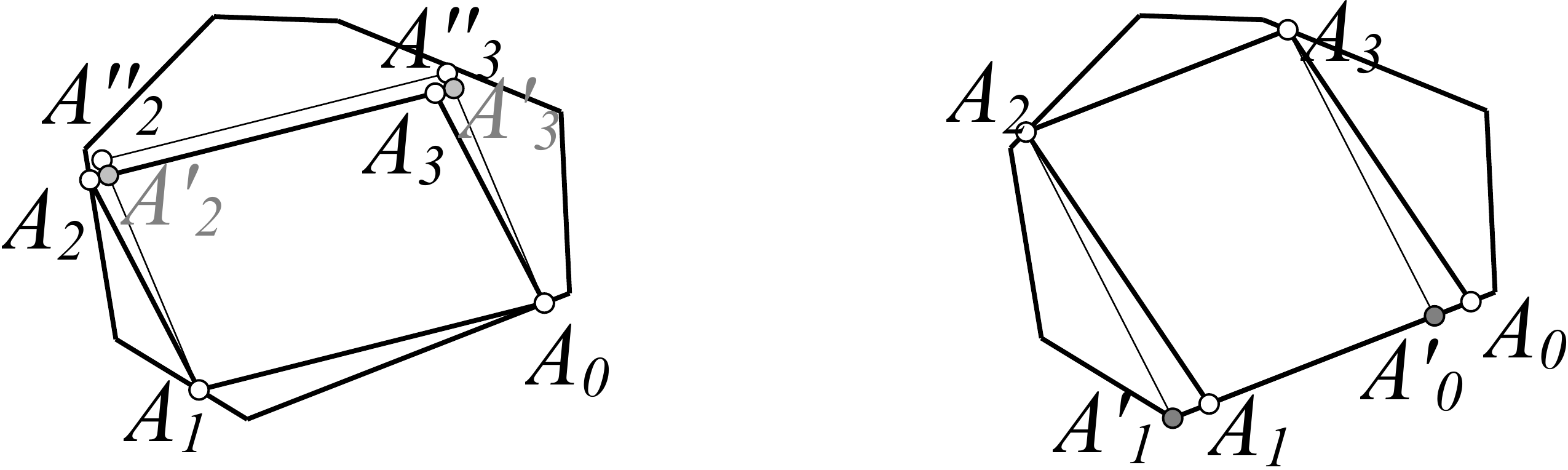}
  \caption{Illustration of the proof of Lemma~\ref{lemma:local-maximal-non-slidable}.}\label{fig:LMAPs_fact}
\end{figure}

\begin{proof}[Proof of Lemma~\ref{lemma:local-maximal-non-slidable}]
1. For a contradiction, suppose parallelogram $Q=A_0A_1A_2A_3$ is locally maximal but not inscribed.
Without loss of generality, assume $A_3\notin \partial P$. See Figure~\ref{fig:LMAPs_fact}~(a).
First, we slide (namely, translate) segment $A_2A_3$ along direction $\overrightarrow{A_2A_3}$ for a sufficiently small distance to create $A'_2A'_3$.
Next we slide it along direction $\overrightarrow{A_0A'_3}$ for a sufficiently small distance to create $A''_2A''_3$ where
$A''_2$ and $A''_3$ are still in $P$.
Clearly, $Area(A_0A_1A_2A_3)<Area(A_0A_1A''_2A''_3)$, so $Q$ is not locally maximal.

2. Assume $Q=A_0A_1A_2A_3$ is inscribed and $A_0,A_1$ lies in the same edge. See Figure~\ref{fig:LMAPs_fact}~(b).
We slide segment $A_0A_1$ along direction $\overrightarrow{A_0A_1}$ to create $A'_0A'_1$ so that $A'_1$ coincides with an endpoint of the edge.
Then, $A'_1$ does not lie on this edge, since the edge does not contain its endpoints.
So the new parallelogram $A'_0A'_1A_2A_3$ is non-slidable. Moreover, it is inscribed and it clearly has the same area as $A_0A_1A_2A_3$.
 \end{proof}

\section{The problem can be formulated as a linearly constrained quadratic yet non-convex programming}\label{sect:address-cp}

The MAP problem can be easily formulated as:
\begin{equation}\label{eqn:plausibleCP}
\begin{small}
\begin{gathered}
  f(x_0,y_0,x_1,y_1,x_2,y_2)=\max \left|\det \left(
                                 \begin{array}{cc}
                                   x_1 & x_2 \\
                                   y_1 & y_2 \\
                                 \end{array}
                               \right) \right| ,
  \text{ s. t.} \left\{\begin{array}{c}
                        (x_0+x_1,y_0+y_1) \in P \\
                        (x_0-x_1,y_0-y_1) \in P \\
                        (x_0+x_2,y_0+y_2) \in P \\
                        (x_0-x_2,y_0-y_2) \in P
                      \end{array}\right.
\end{gathered}
\end{small}
\end{equation}

The absolute function in this formulation can be removed without bringing any difference.
  Therefore, (\ref{eqn:plausibleCP}) is a (linearly constrained) \emph{quadratic programming}.\medskip

It is well-known that $g(X) = \log \det X$ is concave on domain of symmetric positive definite $n\times n$ matrices $S^n_{++}$
(see \cite{ellipse-book-co}).
Based on this result, (\ref{eqn:plausibleCP}) looks like a convex programming - indeed, two professors in my previous university who are also reviewers of my follow-up paper
   wrong thought that (\ref{eqn:plausibleCP}) is a convex programming.
This is not true, since matrix $\left(\begin{array}{cc}x_1 & x_2 \\y_1 & y_2 \\\end{array}\right)$ may not belong to $S^2_{++}$ in our case.

\bigskip \noindent \textbf{Claim.}
\emph{Function $f$ is \textbf{not} concave and so (\ref{eqn:plausibleCP}) is \textbf{not} a convex programming.}

\begin{proof}
 \begin{equation*}
    \text{Let} \left\{ \begin{array}{ccccc}
              \mathbf{x} & = & (x_0,y_0,x_1,y_1,x_2,y_2) & = & (0,0,+1,+1,+1,-1)\\
              \mathbf{x}' & = & (x'_0,y'_0,x'_1,y'_1,x'_2,y'_2) & = & (0,0,-1,-1,-1,+1)
            \end{array} \right.
 \end{equation*}

We get $f(\mathbf{x})=f(\mathbf{x}')=2$. Yet $f(\frac{1}{2}(\mathbf{x}'+\mathbf{x}))=f(\mathbf{0})=0$.
 \end{proof}

\section{An example where there are multiple LMAPs}\label{sect:many-pentagon}

When $P$ is a regular pentagon, using the clamping bounds and the algorithm in section~\ref{sect:algo},
  we see that $P$ admits five LMAPs (which are also MAPs). See Figure~\ref{fig:many}.

\begin{figure}[h]
  \centering
  \includegraphics[width=.65\textwidth]{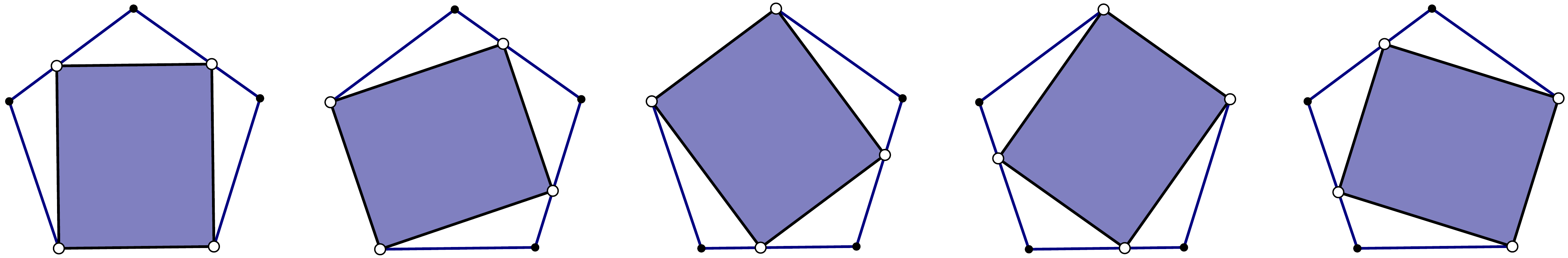}\\
  \caption{The regular pentagon is an example that admits five LMAPs.}\label{fig:many}
\end{figure}

\bibliographystyle{ws-ijcga}
\bibliography{MAP-J}

\begin{thebibliography}{10}
\newcommand{\enquote}[1]{#1}

\bibitem{Classic-shamosCG-dissertation}
M.~Shamos, \emph{Computational geometry} (1978).

\bibitem{Potato-polynomial-DCG86}
J.~Chang and C.~Yap, \enquote{A polynomial solution for the potato-peeling
  problem}, \emph{Disc. \& Comp. Geom.} \textbf{1} (1986) 155.

\bibitem{Kgon-extremal-STOC82}
J.~E. Boyce, D.~P. Dobkin, R.~L.~S. Drysdale, III and L.~J. Guibas,
  \enquote{Finding extremal polygons}, in \emph{14th Symposium on Theory of
  Computing} (1982), STOC '82, pp. 282--289.

\bibitem{ellipse-LPtype-SCG92}
J.~Matou\v{s}ek, M.~Sharir and E.~Welzl, \enquote{A subexponential bound for
  linear programming}, in \emph{8th Proc. of Symp. on Comp. Geom.} (New York,
  NY, USA, 1992), pp. 1--8.

\bibitem{Othershape-square-Allerton87}
N.~Pano, Y.~Ke and J.~O'Rourke, \enquote{Finding largest inscribed equilateral
  triangles and squares}, in \emph{Proc. of Alle. Conf. on Commu., Cont., and
  Comp.} (1987), pp. 869--878.

\bibitem{Placement-chazelle-CR83}
B.~Chazelle, \enquote{The polygon containment problem}, \emph{J. of Adv. in
  Comp. Res.} \textbf{1} (1983) 1.

\bibitem{Placement-ST-CGTA94}
M.~Sharir and S.~Toledo, \enquote{Extremal polygon containment problems},
  \emph{Comp. Geom.: Theor. and Appl.} \textbf{4} (1994) 99.

\bibitem{Placement-convex-DCG98}
P.~K. Agarwal, N.~Amenta and M.~Sharir, \enquote{Largest placement of one
  convex polygon inside another}, \emph{Disc. \& Comp. Geom.} \textbf{19}
  (1998) 95.

\bibitem{ShapeSurvey}
B.~Chazelle, \emph{Approximation and Decomposition of Shapes} (Routledge,
  1987), ch.~0, pp. 145--185, Routledge Revivals.

\bibitem{Math-Sas-CM39}
E.~Sas, \enquote{\"{U}ber ein extremumeigenschaft der ellipsen}, \emph{Compos.
  Math.} \textbf{6} (1939) 468.

\bibitem{Math-Dow-AMS44}
C.~Dowker, \enquote{On minimum circumscribed polygons}, \emph{Bullet. of the
  Amer. Math. Soc.} \textbf{50} (1944) 120.

\bibitem{followup-MAP-arxiv15}
K.~Jin, \enquote{Maximal parallelograms in convex polygons - a novel geometric
  structure}, \emph{CoRR} \textbf{abs/1512.03897}.

\bibitem{Kgon-Matrixsearch-Algc87}
A.~Aggarwal, M.~M. Klawe, S.~Moran, P.~Shor and R.~Wilber, \enquote{Geometric
  applications of a matrix-searching algorithm}, \emph{Algorithmica} \textbf{2}
  (1987) 195.

\bibitem{Kgon-Klink-DCG94}
A.~Aggarwal, B.~Schieber and T.~Tokuyama, \enquote{Finding a minimum-weight
  k-link path in graphs with the concave monge property and applications},
  \emph{Disc. \& Comp. Geom.} \textbf{12} (1994) 263.

\bibitem{Kgon-Klink-soda95}
B.~Schieber, \enquote{Computing a minimum-weight k-link path in graphs with the
  concave monge property}, in \emph{6th Symp. on Disc. Algo.} (1995), SODA '95,
  pp. 405--411.

\bibitem{Triangle-correct-IJCGA02}
S.~Chandran and D.~M. Mount, \enquote{A paraellel algorithm for enclosed and
  enclosing triangles}, \emph{Int. J. of Comp. Geom. \& Appl.} \textbf{02}
  (1992) 191.

\bibitem{Triangle-ultimate-Arxiv}
K.~Jin, \enquote{Maximal area triangles in a convex polygon}, \emph{CoRR}
  \textbf{abs/1707.04071}.

\bibitem{Triangle-wrong-FOCS79}
D.~Dobkin and L.~Snyder, \enquote{On a general method for maximizing and
  minimizing among certain geometric problems}, in \emph{20th Symp. on Foun. of
  Comp. Sci.} (1979), pp. 9--17.

\bibitem{Triangle-reportwrong-Arxiv}
V.~Keikha, M.~L{\"{o}}ffler, J.~Urhausen and I.~v.~d. Hoog,
  \enquote{Maximum-area triangle in a convex polygon, revisited}, \emph{CoRR}
  \textbf{abs/1705.11035}.

\bibitem{Othershape-rect-EuroCG14}
S.~Cabello, O.~Cheong and L.~Schlipf, \enquote{Finding largest rectangles in
  convex polygons}, in \emph{Euro. Workshop on Comp. Geom.} (2014).

\bibitem{Triangle-EnclArea-JA86}
J.~O'Rourke, A.~Aggarwal, S.~Maddila and M.~Baldwin, \enquote{An optimal
  algorithm for finding minimal enclosing triangles}, \emph{J. of Algo.}
  \textbf{7} (1986) 258 .

\bibitem{Triangle-EnclPeri-JCDCG02}
B.~Bhattacharya and A.~Mukhopadhyay, \enquote{On the minimum perimeter triangle
  enclosing a convex polygon}, in \emph{Discrete and Computational Geometry},
  eds. J.~Akiyama and M.~Kano (Springer Berlin Heidelberg, Berlin, Heidelberg,
  2003), pp. 84--96.

\bibitem{Classic-rcaliper-MEL83}
G.~Toussaint, \enquote{Solving geometric problems with the rotating calipers},
  in \emph{IEEE MELECON'83} (1983), pp. 1--4.

\bibitem{3dorEncl-areaParallelogram-SOCG95}
C.~Schwarz, J.~Teich, A.~Vainshtein, E.~Welzl and B.~Evans, \enquote{Minimal
  enclosing parallelogram with application}, in \emph{11th Proc. of Symp. on
  Comp. Geom.} (ACM, New York, NY, USA, 1995), pp. 434--435.

\bibitem{3dorEncl-periParallelogram-cccg10}
Y.~Bousany, M.~Karker, J.~O¡¯rourke and L.~Sparaco, \enquote{Sweeping minimum
  perimeter enclosing parallelograms: Optimal crumb cleanup}, in \emph{22rd
  Proc. of Cana. Conf. on Comp. Geom.} (2010), pp. 167--170.

\bibitem{ellipse-LPtype-JA96}
B.~Chazelle and J.~Matou\v{s}ek, \enquote{On linear-time deterministic
  algorithms for optimization problems in fixed dimension}, \emph{J. of Algo.}
  \textbf{21} (1996) 579.

\bibitem{ellipse-spanconvex-SCG92}
M.~Dyer, \enquote{A class of convex programs with applications to computational
  geometry}, in \emph{8th Proc. of Symp. on Comp. Geom.} (New York, NY, USA,
  1992), pp. 9--15.

\bibitem{ellipse-ipconvex-tr}
Y.~Zhang, \enquote{An interior-point algorithm for the maximum-volume ellipsoid
  problem}, Technical report, RICE UNIVERSITY, 1999.

\bibitem{ellipse-book-co}
S.~Boyd and L.~Vandenberghe, \emph{Convex Optimization} (Cambridge University
  Press, New York, NY, USA, 2004).

\bibitem{ellipse-john-CAV48}
F.~John, \enquote{Extremum problems with inequalities as subsidiary
  conditions}, \emph{Courant Anniversary Volume}  (1948) 187.

\bibitem{Math-approx-GD98}
M.~Lassak, \enquote{Approximation of convex bodies by centrally symmetric
  bodies}, \emph{Geometriae Dedicata} \textbf{72} (1998) 63.

\bibitem{Math-approx-JDG04}
Y.~Gordon, A.~E. Litvak, M.~Meyer and A.~Pajor, \enquote{John's decomposition
  in the general case and applications}, \emph{J. of Diff. Geom.} \textbf{68}
  (2004) 99.

\bibitem{Math-3dinscribed-french54}
A.~Bielecki and K.~Radziszewski, \enquote{Sur les parall\'{e}l\'{e}pip\`{e}des
  inscrits dans les corps convexes}, \emph{Ann. Univ. Mariae
  Curie-Sk{\l}odowska, Sect. A} \textbf{8} (1954) 97.

\bibitem{Math-survey-GM97}
T.~Hausel, E.~Makai and A.~Sz\"{u}cs, \enquote{Polyhedra inscribed and
  circumscribed to convex bodies}, \emph{General Mathematics} \textbf{5} (1997)
  183.

\bibitem{Math-PinC-AMM60}
C.~M. Fulton and S.~K. Stein, \enquote{Parallelograms inscribed in convex
  curves}, \emph{Amer. Math. Month.} \textbf{67} (1960) 257.

\bibitem{Math-3dinscribed-french52}
K.~Radziszewski, \enquote{Sur une probl\'{e}me extr\'{e}mal relatif aux figures
  inscrites et circonscrites aux fiures convexes}, \emph{Ann. Univ. Mariae
  Curie-Sk{\l}odowska, Sect. A} \textbf{6}.

\bibitem{Math-PinS-DM00}
G.~Leng, Y.~Zhang and B.~Ma, \enquote{Largest parallelotopes contained in
  simplices}, \emph{Disc. Math.} \textbf{211} (2000) 111.

\bibitem{Math-PinS-DCG99}
M.~Lassak, \enquote{Parallelotopes of maximum volume in a simplex}, \emph{Disc.
  {\&} Comp. Geom.} \textbf{21} (1999) 449.

\bibitem{Math-PinE-AMM07}
A.~Connes and D.~Zagier, \enquote{A property of parallelograms inscribed in
  ellipses}, \emph{Amer. Math. Month.} \textbf{114} (2007) 909.

\bibitem{Math-SafeDomain-EJP04}
J.~M. Richard, \enquote{Safe domain and elementary geometry}, \emph{Euro. J. of
  Phy.} \textbf{25} (2004) 835.

\bibitem{Heilbronn-convexhull-LMS71}
W.~M. Schmidt, \enquote{On a problem of {H}eilbronn}, \emph{J. of the Lond.
  Math. Soc.} \textbf{4} (1971) 545.

\bibitem{Heilbronn-Lefmann-JC00}
C.~Bertram{-}Kretzberg, T.~Hofmeister and H.~Lefmann, \enquote{An algorithm for
  {H}eilbronn's problem}, \emph{{SIAM} J. on Comp.} \textbf{30} (2000) 383.

\bibitem{Heilbronn-Lefmann3d-JC02}
H.~Lefmann and N.~Schmitt, \enquote{A deterministic polynomial-time algorithm
  for {H}eilbronn's problem in three dimensions}, \emph{{SIAM} J. on Comp.}
  \textbf{31} (2002) 1926.

\end{thebibliography}

\end{document}